\RequirePackage{fix-cm}
\documentclass[twocolumn]{svjour3}          
\smartqed  
\usepackage{graphicx}
\graphicspath{{.}{_img/}}

\usepackage{mathptmx}      
\usepackage[utf8]{inputenc}
\usepackage[T1]{fontenc}
\usepackage{amsmath,amssymb,amsfonts}
\usepackage{tabulary}
\usepackage{booktabs}
\usepackage{multirow}
\usepackage{cite}
\usepackage{xcolor}
\usepackage{tikz}
\usepackage{here}
\newcommand{\mt}{\ensuremath{\mathrm{MT}}}
\newcommand{\W}{\ensuremath{{W}}}  
\newcommand{\D}{\ensuremath{{D}}} 
\newcommand{\id}{\ensuremath{\mathrm{ID}}}

\newcommand{\etal}{et~al.}%

\newif\ifcomments \commentsfalse

\newcommand\defineComment[3]{%
\ifcomments%
\expandafter\newcommand\csname #1\endcsname[1]{{\leavevmode\color#2{#3##1}}}%
\else%
\expandafter\newcommand\csname #1\endcsname[1]{\ignorespaces}%
\fi
}

\newcommand\definePersistentComment[3]{%
\ifcomments%
\expandafter\newcommand\csname #1\endcsname[1]{{\leavevmode\color#2{#3##1}}}%
\else%
\expandafter\newcommand\csname #1\endcsname[1]{{\leavevmode\color#2{#3##1}}}
\fi
}

\definePersistentComment{added}{{black}}{}
\defineComment{removed}{{orange}}{}
\definePersistentComment{moved}{{cyan}}{{\sc moved:}}

\usepackage{microtype}
\usepackage{flushend}
\usepackage{hyperref}

\begin{document}
\def\makeheadbox{{%
\hbox to0pt{\vbox{\baselineskip=10dd\hrule\hbox
to\hsize{\vrule\kern3pt\vbox{\kern3pt
\hbox{ This is a pre-print of an article published in \emph{Biological Cybernetics}. The final authenticated version is available online}
\hbox{with the following DOI: https://doi.org/10.1007/s00422-020-00853-7}
\kern3pt}\hfil\kern3pt\vrule}\hrule}%
\hss}}}

\title{A Feedback Information-Theoretic Transmission Scheme (FITTS) for Modeling Trajectory Variability in Aimed Movements}


\author{Julien Gori         \and
        Olivier Rioul 
}


\institute{Julien Gori \at
              LRI,\\
              Universit\'{e} Paris-Saclay, CNRS, Inria, \\
              F-91400, Orsay, France.\\
              \email{juliengori@gmail.com}           
           \and
           Olivier Rioul \at
              LTCI, T\'{e}l\'{e}com Paris, \\
              Institut Polytechnique de Paris, \\
              F-91120, Palaiseau, France. \\
              \email{olivier.rioul@telecom-paris.fr}
}

\date{Received: date / Accepted: date}

\maketitle

\begin{abstract}
%
Trajectories in human aimed movements are inherently variable. 
Using the concept of positional variance profiles, such trajectories are shown to be decomposable into two phases: 
In a first phase, the variance of the limb position over many trajectories increases rapidly; in a second phase, it then decreases steadily. 
A new theoretical model, where the aiming task is seen as a Shannon-like communication problem, is developed to describe the second phase: 
Information is transmitted from a “source” (determined by the position at the end of the first phase), to a “destination” (the movement’s end- point) over a “channel” perturbed by Gaussian noise, with the presence of a noiseless feedback link. 
Information-theoretic considerations show that the positional variance decreases exponentially with a rate equal to the channel capacity C. 
Two existing datasets for simple pointing tasks are re-analyzed and observations on real data confirm our model. 
The first phase has constant duration and C is found constant across instructions and task parameters, which thus characterizes the participant’s performance. 
Our model provides a clear understanding of the speed-accuracy tradeoff in aimed movements: 
Since the participant’s capacity is fixed, a higher prescribed accuracy necessarily requires a longer second phase resulting in an increased overall movement time. The well-known Fitts’ law is also recovered using this approach.

\end{abstract}

\section{Introduction}
\label{sec:intro}
It has long been observed that people routinely adapt their speed to perform aimed movements in a reliable manner: The more accurate a movement the slower its execution.
This so-called speed-accuracy tradeoff has been studied for more than a century by many communities such as experimental psychology~\cite{woodworth1899, fitts1954, welford1960, crossman1983, meyer1988, crossman1957}, human-computer interaction (HCI)~\cite{card1978, soukoreff2004}, cybernetics~\cite{chan1990a, gawthrop2011}, robotics~\cite{simmons2005},  and neuroscience~\cite{flanagan1995, khan2006}.


In 1954, Fitts~\cite{fitts1954} provided a simple formula to describe the speed-accuracy tradeoff of simple aimed movements (Fitts' law, see \S~\ref{sec:background}), focusing on the variability of the movement \emph{endpoints}.
Fitts' law remains to this day heavily used in various communities, e.g., in HCI to evaluate the performance of input devices~\cite{soukoreff2004}; unfortunately it is only a rule of thumb, initially conceived with a vague analogy with information theory~\cite{gori2018}, that only deals with movement endpoints and says little about the entire movement's trajectory.

Since Fitts' seminal 1954 work, many attempts have been pursued to explain Fitts' law as resulting from an underlying mechanism for movement production and control.
Examples are neural dynamics models~\cite{bullock1988}, behavioral models~\cite{elliott2017, elliott2010}, various mathematically formulated models~\cite{plamondon1997, meyer1988} and models based on control theory, including simple closed-loop step responses~\cite{mueller2017}, deterministic and stochastic optimal control~\cite{todorov2002, tanaka2006, harris1998, flash1985, guigon2007, berret2016}.
This last class of models has been particularly influential in explaining motor planning~\cite{todorov2002}. It shows how humans choose one particular path from the infinitely many possible~\cite{flash1985} (this is usually known as the joint-redundancy or degree of freedom problem~\cite{rosenbaum2009,bullock1988}), and how and when they will correct deviations from that particular path~\cite{todorov2002}.

Control theoretic models do suffer from several difficulties. To be operational, they require estimating the impedance of the motor system, i.e. limb and muscles inertia and visco-elastic properties. However, modulating impedance might be a control strategy in its own right~\cite{hogan1985}, and even estimating an impedance that is assumed invariant is not trivial~\cite{buchanan2004, mueller2017}. Another issue is that it is not clear which cost functions are relevant to the CNS (Central Nervous System), and how they should be combined~\cite{franklin2011} in optimal control models.

In this paper, we circumvent these difficulties by using an optimal information transmission model that explains the evolution of variance over time for precise aimed movements, jointly maximizing speed and accuracy without relying on motor impedance.
We leverage a computational model rooted in information theory to show that the problem of aiming can be seen as a Shannon-like communication problem~\cite{shannon1948} with a noisy feedforward channel and a noiseless feedback link.
Using Shannon's information theory is not only a reminder of the initial rationale behind Fitts' law~\cite{fitts1954, gori2018}; it mostly comes from the observation that noise is ``the one key element limiting motor performance''~\cite{franklin2011}---and that information (communication) theory was conceived with the goal of finding ``ways of transmitting the information which are optimal in combating noise''~\cite{shannon1948}.
In spite of the claimed differences between our information-theoretic and other optimal control schemes, similarities do exist, such as the use of feedback information at the emitter, and the presence of a Minimum Mean Square Error (MMSE) estimator at the receiver. Links between control theoretic and information-theoretic formalisms have been made~\cite{elia2004, tatikonda2000}, and are out of the scope of this paper.

We then introduce \emph{positional variance profiles} (PVPs) as a way to operationalize our model and track the evolution of the variability of trajectories over time.
We remark that PVPs are necessarily \emph{unimodal}: A first distance-covering phase where positional variance increases rapidly over time is followed by a much longer second phase --- which makes reliable aiming possible. Our model explains this longer second phase of decreasing variance, while an empirical study reveals that the duration of the first phase is approximately constant.
We conclude by combining the two phases, which produces a novel derivation for Fitts' law.

The remainder of this paper is organized as follows. Behavioral observations relevant to our model are described in \S~\ref{sec:background}. The information-theoretic model for the second phase is derived in \S~\ref{sec:model}, as the result of a joint optimization of aiming accuracy and movement time. The main theoretical result is that the variance decreases exponentially over time during the second phase. PVPs are described in~\S~\ref{sec:method} and empirical evidence to support our model is provided in \S~\ref{sec:empirical_two}. An empirical study of the first phase of the PVPs with increasing variance is presented in \S~\ref{sec:pvp_results}. Finally, the synthesis of the two phases leads to the recovery of the well-known Fitts' law in \S~\ref{sec:fitts}. Section~\S~\ref{sec:conclusion} concludes.

\section{Background \label{sec:background}}
Our aim in this section is not to present an exhaustive review of motor control models; instead we present behavioral observations on the variability of movements, the use of feedback information and the continuous and discrete nature of the control of movements that are relevant to our model. 
\subsection{Variability of Movements}
Human movement is inherently variable, as recognized since the earliest studies on human produced movement~\cite{woodworth1899}. The variability of the endpoints of movement has been of great interest to behaviorists, leading to our first point of interest, the well-known empirical Fitts' law.
A second point of interest is the observation that the evolution of the variance of position over time in aiming tasks has a unimodal profile, of which an ``ideal'' is represented Fig.~\ref{fig:sig_th}.

\begin{figure}[htbp]
\centering
\includegraphics[width=.9\columnwidth]{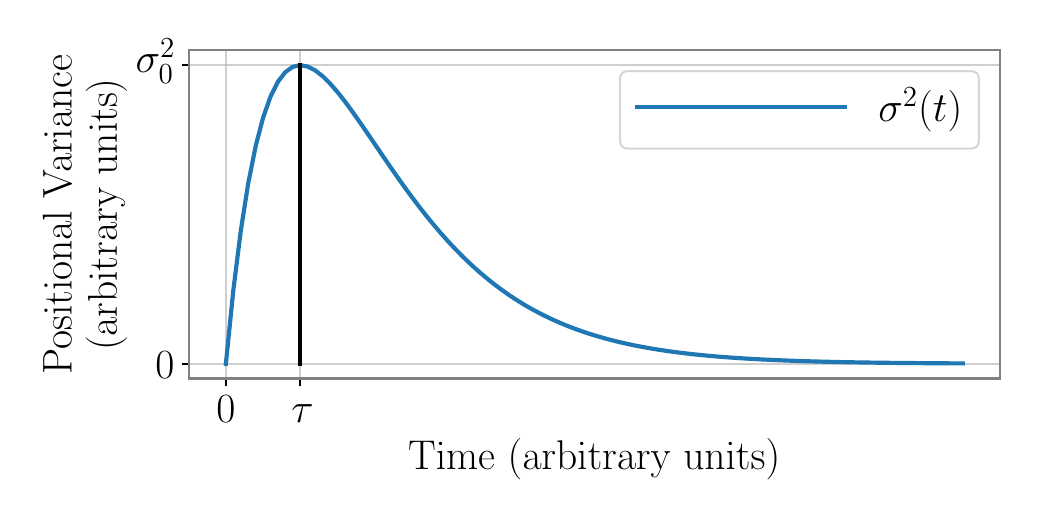} 
\caption{Ideal two-phase positional variance profile.
The transition between the two phases occurs at $(\tau;\sigma_0^2)$.}
\label{fig:sig_th}
\end{figure}

\paragraph{Fitts' law.}
Fitts' law describes the time \mt{} it takes to reach a target of size \W{} located \D{} away:
\begin{align}
\mt{} = a + b\,\log_2 \bigl(1 + \tfrac{\D{}}{\W{}} \bigr) = a + b\,\id{} \qquad \text{(seconds).} \label{eq:fitts}
\end{align}
$\id{} = \log_2 \Bigl(1 + \frac{\D{}}{\W{}} \Bigr)$ is known as the index of difficulty. There is, at present, no way to predict values for $a$ (the intercept, in seconds) and $b$ (the slope, in seconds per bit), which therefore have to be estimated from empirical data~\cite{soukoreff2004, gori2018}. 
The inverse of the slope ($1/b$, in bits per second) is often called throughput~\cite{zhai2004b} and can be interpreted as a measure of the information-theoretic ``capacity'' of the human motor system i.e. the highest rate of information that is transfered during the task completion.

Fitts' law was first stated by Fitts~\cite{fitts1954} as a vague analogy with the Shannon capacity formula~\cite{shannon1948} ---that analogy having been questioned on many occasions~\cite{sheridan1974, gori2018}.
Due to the informal nature of Fitts' theoretical construct, many variations of the law exist\footnote{For example, at least a dozen different formulations for \id{} exist~\cite{plamondon1997}, including Fitts' original $\id{} = \log_2 \bigl( \frac{2\D{}}{\W{}} \bigr)$.} and most of the practices surrounding Fitts' law are empirically guided. 

A variation that is used e.g., in HCI~\cite{soukoreff2004, gori2018}, known as the effective Fitts' law~\cite{zhai2004a}, corrects for participants that miss targets over the course of the experiment, replacing \W{} with the actual standard deviation of endpoints $\sigma$, estimated from the data:
\begin{align}
\mt{} = a + b\,\log_2 \bigl(1 + \tfrac{\D{}}{4.133\sigma}\bigr) = a + b\,\id{}_e \label{eq:fitts_mackenzie}
\end{align}
where $\id{}_e = \log_2 \left(1 + \frac{\D{}}{4.133\sigma}\right)$ is known as the effective index of difficulty.
This formulation is however based on the heuristic that an $\varepsilon = 3.88\%$ miss rate should be pursued and lacks a satisfying justification~\cite{gori2018}.

\paragraph{Unimodal Positional Variance.}
As noticed in~\cite{todorov2002}, trajectory variability has ``surprisingly received little attention from researchers'': Most studies investigating Fitts' law have been limited to the measure of the spread of endpoints~\cite{crossman1957, soukoreff2004} or single movements~\cite{meyer1988, plamondon1997}. 
However, that humans are able to reach almost any target reliably implies that they are able to reduce their endpoint variability at will, which begs the question of how exactly that variability is reduced throughout the trajectory.

The evolution in time of the standard deviation of trajectories from a tapping task was evaluated in~\cite{lai2005} and appeared \emph{unimodal} (an increasing phase was followed by a decreasing phase).
Other studies~\cite{khan2006,vandermeulen1990} have represented positional variance at specific kinematic markers (peak acceleration, peak velocity, peak deceleration, movement time), which also suggests unimodal variance profiles.
Finally, Gutman~\etal{}~\cite{gutman1992, gutman1993} proposed a feedforward model where a primitive trajectory is rescaled in time and space depending on the task. They found their model predicted unimodal variance profiles, in line with empirical observations. \footnote{They did find bimodal profiles, and so did Darling~\etal{}~\cite{darling1987a} in another study, on fast elbow only (single-joint) flexions. This differs from our work which tackles multi-joint movements).}
Thus, evidence suggests that the positional variance has a unimodal profile in time; this was also verified in our empirical analysis (see \S~\ref{sec:pvp_results}).

The model that we propose accounts for the variance decreasing phase of movement: we namely show that the variance, in the decreasing phase, does so at a characteristic exponential rate.

\subsection{Feedforward and Feedback Information}
It is well known that humans cannot function properly without feedback, see e.g. Wiener's account~\cite[p.95]{wiener1948} on two patients suffering from a lack of voluntary muscle movement coordination known as \textit{ataxia} caused by a lack of proprioceptive feedback information at the brain. Another illustration comes from the fact that older adults have a diminished sense of proprioception which seems to account for a decrease in their performance compared with younger adults~\cite{ketcham2004}. 
Movement generation also relies on visual feedback mechanisms: performing precise movements with closed eyes is near impossible, and various experiments on occlusion and removal of light~\cite{elliott1995, zelaznik1983} or removal of cursor~\cite{chua1993} confirm an effect of visual feedback on virtually all kinematic properties, including accuracy and movement time, on movements as short as 100~ms~\cite{elliott2010,carlton1992,zelaznik1983}. 

However, many works also indicate feedforward control~\cite{mehta2002}, in particular those with deafferented monkeys~\cite{polit1978, polit1979, bizzi1984}.
Incorporating both feedforward and feedback mechanisms is thus necessary when designing a scheme to explain aimed movement production~\cite{wolpert1997}. 
This was already recognized in earliest descriptions of aimed movements by Woodworth~\cite{woodworth1899}, who hypothesized that an aimed movement is constituted of two serial components: an initial adjustment, whose main purpose is to cover distance, followed by a homing-in phase that relies on vision to ensure accuracy. 
This two-component model has been recently refined under the impulse of Elliott and colleagues~\cite{elliott2001, elliott2010, elliott2017}:
A first \emph{planned} component gets the limb close to the target area. This component is based on internal representations and is associated to a velocity regulation through proprioceptive feedback. Typically high speeds and low jerk trajectories are associated with this component.
When time permits, a second corrective portion is engaged to reduce spatial discrepancy between limb and target. This process, highly dependent on foveal (central) vision, involves computing the difference between limb and target position and issuing discrete corrections. 

\textcolor{black}{Our model incorporates feedback and feedforward mechanisms via a feedback and a feedfoward channel. We will also find many resemblances between the two aforementioned phases of variance and the two components of Woodworth \& Elliott.}

\subsection{Continuous, Discrete, and Intermittent control}
Woodworth~\cite{woodworth1899} was the first to suggest that aiming was comprised of a first distance-covering phase followed by a second homing-in phase, but the segmentation was never explicitly performed, including in its modern exposition~\cite{elliott2001, elliott2010}.
It was later observed that the velocity during movement vanished multiple times before its end, which was interpreted as the transition from one submovement to another---the whole movement being composed of concatenated submovements---which implies a discrete type of control~\cite{crossman1983, keele1968, meyer1988}.
\textcolor{black}{Several models have used this idea of submovements, including Crossman and Goodeve's~\cite{crossman1983} deterministic iterative control model, and its stochastic counterpart due to Meyer~\etal{}~\cite{meyer1988}, which have proved quite popular.}

There are however many reports of improvement of overall performance metrics in the presence of feedback, even though no distinct changes occur in the kinematic profiles~\cite{pelisson1986,desmurget2000, elliott1991}. This would indicate continuous control, and many successful optimal control models are indeed continuous.
Yet another form of control is intermittent control, where e.g. observations are continuous, but actions are taken intermittently~\cite{gawthrop2011}.

It has been suggested that training and learning leads to the transition from discrete or intermittent control to continuous control~\cite{woodworth1899,proteau1987,elliott2010}. It has also been suggested that intermittent control can appear to be continuous---so-called masquerading~\cite{gawthrop2011}. This can lead to difficulties with models that form predictions that are based on intermittent/discrete style control strategies, e.g. if the quantity (e.g. amplitude, frequency) or quality of submovements are predicted. Our use of variance profiles, that relies on statistical averaging, allows us to form predictions regardless of the actual type of control.

\section{Variance-Decreasing Phase: \added{Theoretical} Model}
\label{sec:model}
Contrary to most models of movement, where initial position is given (e.g., the optimal control theoretic framework~\cite{todorov2002, tanaka2006, harris1998, flash1985, guigon2007, berret2016}), the input to our model is a \emph{random} variable, corresponding to the position when positional variance is maximized (time $t = \tau$ in Fig.~\ref{fig:sig_th}). In this section, some theorems and lemmas are given. All full proofs are delayed to Appendix~\ref{app:proofs}.

\begin{figure*}
\centering
\begin{tikzpicture}[scale=.9]
\draw[->] (1,0) -- (2,0);
\draw (1,0) node[above]{$\mathbf{A}$};
\draw (2,1) rectangle (6,-1);
\draw (4,0) node[below]{Brain};
\draw (4,0) node[above]{ $\mathbf{f(A,\widehat{A}^{i-1})}$};

\draw[->] (6,0) -- (8.5,0);
\draw (7,0) node[above]{$\mathbf{X_i}$};
\draw (9,0) circle(.5);
\draw (9,0) node{\Large $+$};
\draw[->] (9,1) -- (9,.5);
\draw (9,1) node[right]{$\mathbf{Z_i}$};
\draw[->] (9.5,0) -- (12,0);
\draw (11,0) node[above]{$\mathbf{Y_i}$};

\draw[->] (11,-1.5) -- (4,-1.5) -- (4,-1);
\draw (4,-1.5) node[left]{$\mathbf{\widehat{A}_{i-1}}$};
\draw[->] (15,0) -- (15,-1.5) -- (11,-1.5);
\draw[fill = white] (8,-1.15) rectangle (11,-1.75);
\draw (9.5,-1.5) node{Sample Delay $T$};
\draw (12,1) rectangle (14,-1);
\draw (13,.5) node[]{$\mathbf{g(Y^i)}$};
\draw (13,-.2) node[below]{Muscles};
\draw[->] (14,0)--(16,0);
\draw (15,0) node[above]{$\mathbf{\widehat{A}_i}$};
\end{tikzpicture}
\\

\begin{tikzpicture}[scale = .66]
\draw[->] (1,0)--(3,0);
\draw (1.5,0) node[above]{$\mathbf{A}$};
\draw (3.5,0) circle(.5);
\draw (3.5,0) node{$\Sigma$};
\draw[dotted, black, ultra thick] (2.9,.8) -- (2.9,-.6) -- (6.2,-.6) -- (6.2,.8) -- cycle;
\draw (4.55,.8) node[above]{$\mathbf{f}$};
\draw (4.6,0) node[above]{$\mathbf{A}_i$};
\draw (3.5,-.65) node[below right]{$\mathbf{\widehat{A}_{i-1}}$};
\draw (3.5,-.5) node[below left]{$-$};
\draw (7,0) node[above ]{$\mathbf{X}_i$};
\draw (3,0) node[above left]{$+$};
\draw[->] (4,0) --(7.5,0);
\draw (8,0) circle(.5);
\draw (8,0) node{$\Sigma$};
\draw (7.5,0) node[below left]{$+$};
\draw[->] (8,1.5)--(8,.5);
\draw (8,.5) node[above left]{$+$};
\draw (8,1.5) node[right]{$\mathbf{Z}_i$};
\draw[->] (8.5,0)--(12,0);
\draw (12,0) node[above]{$\mathbf{Y}_i$};
\draw[->] (11,0) -- (11 ,-2) -- (9,-2);
\draw (8.5,-2) circle (.5);
\draw (8.5,-2) node{$\Sigma$};
\draw (9,-2) node[above right]{$+$};
\draw (8.5,-2.5) node[below right]{$+$};
\draw (8,-2) -- (6,-2);
\draw (6,-1.5) rectangle (4.5,-2.5);
\draw[->] (5.25,-1)--(5.25,-1.5);
\draw (5.25,-1.25) node[right]{$T$};
\draw (5.25,-2) node{delay};
\draw[->] (4.5,-2) -- (3.5,-2) -- (3.5,-.5);
\draw[->] (3.5,-2) -- (3.5,-3.5) --(8.5,-3.5)--(8.5,-2.5);
\draw (7,-2) node[above]{$\widehat{\mathbf{A}}_i$};
\draw (11,-2) node[below]{$\mathbb{E}[\mathbf{A}_i|\mathbf{Y}_i]$};
\draw[fill=white] (9,-.5) rectangle (13,-1.5);
\draw (11,-1) node{$(\alpha_i (1+N/P))^{-1}$};

\draw[dotted, black, ultra thick] (13.2,-.3) -- (8.8,-.3) -- (8.8, -1.4) -- (7.8, -1.4) -- (7.8,-3.6) -- (13.2,-3.6) -- cycle;
\draw (13.2,-3.6) node[above left]{$\mathbf{g}$};
\draw[fill=white] (5,.5)rectangle(6,-.5);
\draw (5.5,0) node{$\alpha_i$};

\begin{scope}[xshift = 13cm]
\draw[->] (1,0)--(3,0);
\draw (1.5,0) node[above]{$\mathbf{A}$};
\draw (3.5,0) circle(.5);
\draw (3.5,0) node{$\Sigma$};
\draw (4.6,0) node[above]{$\mathbf{A}_{i}$};
\draw[dotted, ultra thick] (2.9,.6) rectangle (4.1, -.6);
\draw (2.9,.6) node[above right]{Eye};
\draw (3.5,-.65) node[below right]{$\widehat{\mathbf{A}}_{i-1}$};
\draw (3.5,-.5) node[below left]{$-$};
\draw (6.5,0) node[above ]{$\mathbf{X}_i$};
\draw (3,0) node[above left]{$+$};
\draw[->] (4,0) --(7,0);
\draw (7.5,0) circle(.5);
\draw (7.5,0) node{$\Sigma$};
\draw (7,0) node[below left]{$+$};
\draw[->] (7.5,1.5)--(7.5,.5);
\draw (7.5,.5) node[above left]{$+$};
\draw (7.5,1.5) node[right]{$\mathbf{Z}_i$};
\draw[->] (8,0)--(12,0);
\draw (12,0) node[above]{$\mathbf{Y}_i$};
\draw[dotted, ultra thick] (6.1,1.8) rectangle (13.2,-1);
\draw(12,1.8) node[above left]{Nerves};
\draw[dotted, ultra thick] (13.2,-1) rectangle (7,-4);
\draw (12.5,-4) node[above left]{Motor organs};
\draw[dotted, ultra thick] (6.1,1.4) -- (2.3,1.4) -- (2.3,-3) -- (4, -3) -- (4,-1.4) -- (4.7, -1.4) -- (4.7, -.7) -- (6.1,-.7) -- cycle;
\draw (2.3,1.4) node[above right]{Brain};
\draw[->] (11,0) -- (11 ,-2) -- (9,-2);
\draw (8.5,-2) circle (.5);
\draw (8.5,-2) node{$\Sigma$};
\draw (9,-2) node[above right]{$+$};
\draw (8.5,-2.5) node[below right]{$+$};
\draw (8,-2) -- (6,-2);
\draw (6,-1.5) rectangle (4.5,-2.5);
\draw[->] (5.25,-1)--(5.25,-1.5);
\draw (5.25,-1.25) node[right]{$T$};
\draw (5.25,-2) node{delay};
\draw[->] (4.5,-2) -- (3.5,-2) -- (3.5,-.5);
\draw[->] (3.5,-2) -- (3.5,-3.5) --(8.5,-3.5)--(8.5,-2.5);
\draw[fill=white] (9,-.5) rectangle (13,-1.5);
\draw (11,-1) node{$(\alpha_i (1+N/P))^{-1}$};
\draw[fill=white] (5,.5)rectangle(6,-.5);
\draw (5.5,0) node{$\alpha_i$};

\draw (7.5,-2) node[above]{$\widehat{\mathbf{A}}_i$};
\draw (11,-2) node[below]{$\mathbb{E}[\mathbf{A}_i|\mathbf{Y}_i]$};
\end{scope}

\end{tikzpicture}

\caption{Top panel: Information-theoretic model for the aiming task. Bottom left panel: Implementation of the model. Bottom right panel: Physiological representation.}
\label{fig:it_model}
\end{figure*}

\subsection{Notations}
\begin{itemize}
\item bold letters $\mathbf{X}$ and $\mathbf{X}_i$ denote random variables which can be indexed by $i$ in the list $\mathbf{X}^n = (\mathbf{X}_1, \mathbf{X}_2, \dots{}, \mathbf{X}_i, \dots{}, \mathbf{X}_n)$;
\item $\mathcal{N}(\mu, \sigma^2)$ is the $\mu$-centered Gaussian distribution with variance $\sigma^2$;
\item $\mathbb{E}[\mathbf{X}]$ and $\mathbb{E}[\mathbf{X}|\mathbf{Y}]$ are respectively the mathematical expectation of $\mathbf{X}$ and the conditional expectation of $\mathbf{X}$ given $\mathbf{Y}$;
\item $H(\mathbf{X})$ and $H(\mathbf{X}|\mathbf{Y})$ are respectively the differential entropy and the differential conditional entropy, defined by $H(\mathbf{X}) = -\mathbb{E}[\log_2 p(\mathbf{X})]$ and $H(\mathbf{X}|\mathbf{Y}) = -\mathbb{E}[\log_2 p(\mathbf{X}|\mathbf{Y})]$ and are expressed in bits, were $p(\mathbf{X})$ and $p(\mathbf{X}|\mathbf{Y})$ are the probability density functions (pdf) of $\mathbf{X}$ and $\mathbf{X}|\mathbf{Y}$. An important example is the Gaussian entropy~\cite{cover2012}: If $\mathbf{X} \sim \mathcal{N}(0,\sigma^2)$, then $H(\mathbf{X}) = \frac{1}{2}\log_2(2\pi e \sigma^2)$.
\end{itemize}

\subsection{Information-Theoretic Model Description}
Our goal is to explain the reliability of human aimed movements, i.e. how the inherent variability of a set of trajectories gets reduced over time (from $\tau$, the time of the maximum of the PVP, onwards to the end of the movement). In contrast, existing models of goal-directed movements usually predicting trajectories formed between and initial position and final position.

The position of the limb extremity (`limb' in short) at the end of the first phase is \emph{random}: the distance to the target is modeled by the random variable $\mathbf{A}$.

Thanks to the feedback (visual and proprioceptive~\cite{elliott2010}), the limb position is known at the brain level. 
Due to eye-hand coordination and fast eye dynamics, the eye is usually pointing towards the target long before the end of the movement~\cite{elliott2010}. 
Hence we assume that $\mathbf{A}$ is known at the brain level and can be easily evaluated.\footnote{In fact it can be estimated by the eye if the limb is close enough to the target. This is indeed the case since most of the distance to the target has been covered during the first phase, see \S~\ref{sec:pvp_results}.}

We take it that $\mathbf{A}$ is a centered Gaussian with variance $\sigma^2_0$: $\mathbf{A} \sim \mathcal{N}(0,\sigma^2_0)$. We consider $\mathbf{A}$ to be centered, since movements can undershoot and overshoot the target. In general movements can come from any direction (left or right in this 1D model), hence an overshoot for one is an undershoot for the other, symmetrizing $\mathbf{A}$. Empirical evidence in \S~\ref{sec:empirical_two} will further support the Gaussian hypothesis. The first iteration of the scheme is now explained

\paragraph{First iteration of the scheme.}
The second phase of movement is commensurate to sending $\mathbf{A}$, the remaining algebraic distance to the target, from the brain to the limb (to be read with the top panel of Fig.~\ref{fig:it_model}):
\begin{itemize}
\item Given $\mathbf{A}$, the brain outputs an amplitude $\mathbf{X}_1$ to be sent to the limb:
\begin{align}
\mathbf{X}_1 = \mathbf{f}(\mathbf{A}),
\end{align}
where $\mathbf{f}$ is a \emph{deterministic}, yet unknown function performed by the brain, known as the \emph{encoder} in the vocabulary of information theory.
\item The variability of the human motor system is accounted for by a noisy (Gaussian) transmission from brain to limb. The output of the channel $\mathbf{Y}_1$ reads
\begin{align}
\mathbf{Y}_1 = \mathbf{X}_1 + \mathbf{Z}_1,
\end{align} 
where $\mathbf{Z}_1\sim \mathcal{N}(0,N)$ is the ``noise''.
The noise is taken as additive because the remaining distance to the target is small, hence the noise is largely independent from $\mathbf{X}_1$, see~\cite{todorov1998}. This transmission model is the well-known Additive White Gaussian Noise (AWGN) channel~\cite{shannon1948}.
\item The actual distance covered by the limb $\widehat{\mathbf{A}}_1$ is the result of some yet unknown function $\mathbf{g}$ (the \emph{decoder} in the vocabulary of information theory), applied by the motor organs to the received $\mathbf{Y}_1$:
\begin{align}
\widehat{\mathbf{A}}_1 = g(\mathbf{Y}_1).
\end{align}
\item $\widehat{\mathbf{A}}_1$ is returned to the brain \textit{via} noiseless\footnote{We assume that in stereotypical controlled experimental tasks such as those used in this paper, there are no perception issues. A pathology could however be simulated by e.g. introducing on purpose a strong deterioration of the quality of the feedback link.} feedback (e.g. visual, proprioceptive) and used together with $\mathbf{A}$ to produce a new input:  $\mathbf{X}_2 = f(\mathbf{A}, \widehat{\mathbf{A}}_1)$ --- we next show that the best strategy is to use a comparison between $\mathbf{A}$ and $\widehat{\mathbf{A}}_1$ as the subsequent signal.
\end{itemize}
The scheme then progresses iteratively, each step having a constant duration of $T$ seconds. This puts the limiting factor for time not on the dynamics of moving the limb, but on the delay associated with the feedback transmission.\footnote{It is equivalent to consider that all components have zero-delay except for the feedback link, or that the delay is spread over all components. For example, the delay taken into account for the feedback link encompasses the delay associated with the feedforward transmission.}

\paragraph{General form of the iterative scheme.}
At iteration $i \in \lbrace{ 1,\dots,n\rbrace}$, the scheme is described by the following equations (see the top panel of Fig.~\ref{fig:it_model}):
	\begin{enumerate}
		\item 
		The encoder ($\mathbf{f}$) produces $\mathbf{X}_i$ from $\mathbf{A}$ and all received feedback information $\widehat{\mathbf{A}}^{i-1} = \lbrace \widehat{\mathbf{A}}_{1},\widehat{\mathbf{A}}_{2},\dots{},\widehat{\mathbf{A}}_{i-1} \rbrace$:
		\begin{align}
		\mathbf{X}_i = f(\mathbf{A},\widehat{\mathbf{A}}^{i-1}). \label{eq:fx}
		\end{align}
		\item 
		The AWGN channel:
		\begin{align}
		\mathbf{Y}_i = \mathbf{X}_i + \mathbf{Z}_i. \label{eq:awgn}
		\end{align}
		\item 
		The output ($\widehat{\mathbf{A}}_i$) of the decoder ($\mathbf{g}$)  is a function of all received channel outputs $\mathbf{Y}^i = \lbrace \mathbf{Y}_1,\mathbf{Y}_2,\dots{},\mathbf{Y}_i \rbrace$ thus far:
		\begin{align}
		\widehat{\mathbf{A}}_i = \mathbf{g}(\mathbf{Y}^i).\label{eq:gy}
		\end{align}
	\end{enumerate}

The functions $\mathbf{f}$ and $\mathbf{g}$ are assumed to be deterministic and are causal by definition; their exact form however has yet to be specified.
To actually use all previously received signals since the start of the movement, $\mathbf{f}$ and $\mathbf{g}$ require access to some form of memory. Although this is plausible at the brain level for $\mathbf{f}$~\eqref{eq:fx}, it seems less so at the limb level for $\mathbf{g}$~\eqref{eq:gy}. We next determine $\mathbf{f}$ and $\mathbf{g}$ by solving an optimization problem, and we show that no memory is actually required: the optimal solution ---even when assuming the possibility of memorization--- is memoryless.

In Shannon's communication-theoretic terms, the aiming task in the second phase can be seen as the transmission of a real value (distance from target at the end of the first phase) from a ``source'' (CNS) to a ``destination'' (limb extremity) over a noisy forward Gaussian channel with noiseless feedback.
In human-centered terms, that phase makes sure that the limb \emph{reliably} reaches the target, once most of the distance has been covered.

\subsection{Bounds on Transmitted Information}
\label{sub:bounds}

We now leverage information-theoretic definitions.
\begin{itemize}
\item $P_i = \mathbb{E}[\mathbf{X}_i^2]$ is the average power (variance) of the channel input at iteration $i$.
\item The quadratic \emph{distortion} $D_i = \mathbb{E}[(\mathbf{A}-\widehat{\mathbf{A}}_i)^2]$ is the mean-squared error of the estimation of $\mathbf{A}$ by $\widehat{\mathbf{A}}_i$ after $i$ iterations.
\item $I(\mathbf{A};\widehat{\mathbf{A}}_i) = H(\mathbf{A}) - H(\mathbf{A}|\widehat{\mathbf{A}}_i)$ is Shannon's mutual (transmitted) information~\cite{cover2012} between $\mathbf{A}$ and $\widehat{\mathbf{A}}_i$.
\item Shannon's capacity $C$~\cite{shannon1948,cover2012} of the forward AWGN Channel (for one channel use) under the power constraint $P_i \leq P$ and with noise power $N$ is 
\begin{align}
C = \frac{1}{2}\log_2(1 + P/N) \mbox{ (bit per channel use).}\label{eq:capacity}
\end{align}
\end{itemize} 

Shannon's capacity formula is a corrolary to the channel coding theorem~\cite{shannon1948} (for AWGN channels),
which expresses a compromise between a certain measure of speed (rate of information) and a certain measure of accuracy (reliability of a transmission). Goal-directed movements, when modeled according to Fig.~\ref{fig:it_model}, are shown to entail a similar communication tradeoff through the following theorem.
\begin{theorem}
\label{thm:ineq}
Consider the noisy transmission scheme with noiseless feedback of Fig.~\ref{fig:it_model}. For a zero-mean Gaussian source $\mathbf{A}$ with variance $\sigma_0^2$, we have after $n$ iterations
\begin{align}\label{eq:opti_inequal}
\frac{1}{2}\log \frac{\sigma_0^2}{D_n} \; \underset{(a)}{\leq} \; I(\mathbf{A},\widehat{\mathbf{A}}_n) \; \underset{(b)}{\leq} \; nC.
\end{align}
\end{theorem}
\noindent 
The inequality on the left expresses the minimum amount of information that needs to be transmitted from the brain to the limb to reduce the positional variability from the initial variance ($\sigma_0^2$) to the variance after $n$ iterations ($D_n$). 
The inequality on the right expresses the maximum amount of information that can be transmitted over the $n$ times fully exploited noisy channel ($nC$).
Since the transmitted information per channel use can never exceed $C$, and since being more accurate requires sending larger amounts of information, more iterations of the scheme (hence more time) are needed for more precise tasks --- in line with the speed-accuracy tradeoff.

For a given channel capacity $C$ and a given number $n$ of iterations of the scheme, maximizing accuracy is equivalent to minimizing the trajectory variance $D_n$. Similarly, for a given accuracy $D_n$, minimizing time is equivalent to minimizing $n$. Optimal aiming, which we define by those movements that achieve the best possible accuracy in the least amount of time is thus achieved when equality holds in~\eqref{eq:opti_inequal}, i.e. when:
\begin{align}
\frac{1}{2}\log \frac{\sigma_0^2}{D_n} \; = \; I(\mathbf{A},\widehat{\mathbf{A}}_n) \quad = \quad nC. \label{eq:opti_equal}
\end{align} 
The next subsection determines the conditions under which this maximal exploitation of the channel is actually reached.

\subsection{Optimal Aiming (Achieving Capacity)}
\label{sec:achieving}
\begin{lemma}
\label{lm:Cond}
Optimal aiming can be achieved if, and only if, the following conditions hold:
\begin{enumerate}
\item \label{enum:cn:1} all considered random variables $\mathbf{A}$,$\widehat{\mathbf{A}}^i$, $\mathbf{A} - \widehat{\mathbf{A}}^i$,  $\mathbf{X}^i$, $\mathbf{Y}^i$, $\mathbf{Z}^i$ are Gaussian;
\item \label{enum:cn:2} all input powers are equal (to, say, $P$): \newline $P_i = \mathbb{E}[\mathbf{X}_i^2] = P, \quad  \forall i$;
\item \label{enum:cn:3} endpoints $\widehat{\mathbf{A}}^i$ are mutually independent;
\item \label{enum:cn:3_bis} channel outputs $\mathbf{Y}^i$ and errors $\mathbf{A} - \widehat{\mathbf{A}}^i$ are independent;
\item \label{enum:cn:4} $\widehat{\mathbf{A}}_i = \mathbf{g}(\mathbf{Y}^i)$ is a sufficient statistic of $\mathbf{Y}^i$ for $\mathbf{A}$. 
\end{enumerate}
\end{lemma}
\noindent 

Optimal aiming can only be achieved if all the variables are Gaussian (see~\S~\ref{sec:empirical_two} for an empirical investigation).
This Gaussianity suggests that $\mathbf{f}$ and $\mathbf{g}$ are linear functions, as the family of Gaussian distributions is closed under linear combinations. For Gaussian variables, independence is equivalent to decorrelation (or orthogonality), a property we take advantage of in the proofs.

Orthogonality between channel outputs and errors, and between updates (Conditions~\ref{enum:cn:3}. and~\ref{enum:cn:3_bis}. of Lemma 1) suggest MMSE estimation (similarly to the Kalman filter~\cite{kailath1980} heavily used in the stochastic optimal control setting). 
Bridging this result with the previous, it is well-known that the MMSE estimator---which minimizes $D_n$--- is a linear function in the Gaussian setting. 
The general structure of the scheme is thus expected to be linear, with an MMSE estimator on the receiver side.

Since $\mathbf{g}(\mathbf{Y}^i)$ is a sufficient statistic of $\mathbf{Y}_i$ for $\mathbf{A}$, it does not matter if the feedback comes from the endpoints $\widehat{\mathbf{A}}^i$ or from the outputs of the channel $\mathbf{Y}^i$. This is particularly relevant in the present context, since visual and proprioceptive motor feedback information can originate from several sources in the human body.

\bigskip 
 
We now assume for the remainder of the paper that movements are optimal i.e. conditions of Lemma~1 hold.
By further working out the conditions of Lemma~1, we can derive the structure of $\mathbf{f}$ and $\mathbf{g}$. 
We first obtain $\mathbf{g}$ by using the well-known \emph{orthogonality principle}: if $\mathbf{A}$ is to be estimated from the observed data $\mathbf{Y}^i$ by the unbiased estimator $\hat{\mathbf{A}}(\mathbf{Y}^i)$ then the following statements are equivalent:\\[.3ex]
\;\textbullet\;	$\hat{\mathbf{A}}(\mathbf{Y}^i) = \mathbb{E}[\mathbf{A}|\mathbf{Y}^i] = \mathbb{E}[\mathbf{A}\mathbf{Y}^i]^t\mathbb{E}[\mathbf{Y}^i(\mathbf{Y}^i)^t]^{-1}\mathbf{Y}^i$ is the 
	MMSE estimator;\\
\;\textbullet\;     $\mathbb{E}\left[ (\mathbf{A} - \hat{\mathbf{A}}(\mathbf{Y}^i)) \mathbf{Y}_i \right] = 0,~~\forall i$.

\noindent From condition~\ref{enum:cn:3_bis}) we have that $\mathbb{E}[(\mathbf{A}-\mathbf{g}(\mathbf{Y}^i))\mathbf{Y}_i] = 0,~~\forall i $, hence the following theorem results from a direct application of the orthogonality principle.
\begin{theorem}
\label{thm:g}
For the optimal transmission scheme, $\mathbf{g}(\mathbf{Y}^i)$ is the MMSE estimator: $\mathbf{g}(\mathbf{Y}^i) = \mathbb{E} \left[\mathbf{A}|\mathbf{Y}^i \right]$.
\end{theorem}

The optimal scheme thus yields an endpoint $\widehat{\mathbf{A}}_i = \mathbf{g}(\mathbf{Y}^i)$ obtained as the best least-squares estimation of $\mathbf{A}$ from all the current observations of the (independent) channel outputs $\mathbf{Y}^i=(\mathbf{Y}_1,\ldots,\mathbf{Y}_i)$.

\begin{theorem}
\label{thm:f}
For the optimal transmission scheme, $\mathbf{f}$ produces a scaled version of the estimation error:
\begin{align}
\mathbf{X}_i & = \mathbf{f}(\widehat{\mathbf{A}}^{i-1}, \mathbf{A})  = \alpha_i (\mathbf{A} - \widehat{\mathbf{A}}_{i-1}) \\
& = \alpha_i (\mathbf{A} - \mathbb{E}\left[\mathbf{A}|\mathbf{Y}^{i-1}\right]) = \alpha_i (\mathbf{A} - \mathbf{g}(\mathbf{Y}^{i-1})),
\end{align}
 where $\alpha_i$ is such that the power constraint $\mathbb{E}[\mathbf{X}_i^2] = P$ is met.
\end{theorem}
\noindent 

The signal sent at the input of the channel is thus simply the difference between the initial \emph{message} $\mathbf{A}$ and its most recent \emph{estimate} $\widehat{\mathbf{A}}_{i-1}$, rescaled to meet the power constraint. 
The previous two theorems formally define the encoding function $\mathbf{f}$ and decoding function $\mathbf{g}$.

The encoding function $\mathbf{f}$ is mathematically simple and appears biologically feasible, since the difference between $\mathbf{A} - \widehat{\mathbf{A}}_{i-1}$ is simply the remaining distance to the target, which can be estimated easily by the eye. In addition, scaling processesses (for the $\alpha_i$'s) have been identified in the literature e.g. within the basal ganglia~\cite{rosenbaum2009}.
The decoding function $\mathbf{g}$ is expressed as a function of $\mathbf{Y}^i$, which suggests that the motor organs memorize all channel outputs for later use. To our knowledge, this seems unreasonable.
Fortunately, as we next show, there is no need for a memory within the motor organs.

\begin{theorem}
\label{thm:incr}
Let $\mathbf{A}_i = \mathbf{X}_i/\alpha_i$ be the unscaled version of $\mathbf{X}_i$, with $\mathbf{A}_1 = \mathbf{A}$ and $\alpha_1 = 1$. We have:
\begin{align}
\mathbb{E}[\mathbf{A}|\mathbf{Y}^{i}] & = \sum_{j=1}^{\smash{i}} \mathbb{E}[\mathbf{A}|\mathbf{Y}_j] = \sum_{j=1}^{\smash{i}}\mathbb{E}[\mathbf{A}_j|\mathbf{Y}_j] \label{eq:thmincr:1}\\
 & = \sum_{j=1}^{\smash{i}} \frac{1}{\alpha_i} (1 + N/P)^{-1}\mathbf{Y}_j
\end{align}
\end{theorem}
\noindent 
The theorem, through~\eqref{eq:thmincr:1}, shows that the decoding process is recursive: At each step, a ``message'' $\mathbf{A}_i$ that is independent from the previous ones ($\lbrace{\mathbf{A}_1,\ldots,\mathbf{A}_{i-1} \rbrace}$) is formed, sent to the channel, and estimated optimally by least-square minimization: 
Say $\mathbf{A}$ is initially sent over the channel. The first estimate at the decoder is $\mathbb{E}[\mathbf{A}|\mathbf{Y}_1]$; the motor organs respond and the limb moves by $\widehat{\mathbf{A}}_1 = \mathbb{E}[\mathbf{A}|\mathbf{Y}_1]$. The remaining distance is $\mathbf{A}_2 = \mathbf{A} - \mathbb{E}[\mathbf{A}|\mathbf{Y}_1]$, and the next estimate is $\mathbb{E}[\mathbf{A}_2|\mathbf{Y}_2]$. The limb moves again, by $\widehat{\mathbf{A}}_2 = \mathbb{E}[\mathbf{A}_2|\mathbf{Y}_2]$. The remaining distance is now $\mathbf{A}_3 = \mathbf{A} - \left(\mathbb{E}[\mathbf{A}|\mathbf{Y}_1] + \mathbb{E}[\mathbf{A}_2|\mathbf{Y}_2] \right)$.
The estimate at each iteration is thus simply the sum of a new estimate and the sum of all previous estimates i.e. exactly the total distance covered since $\mathbf{A}$. This distance, and hence the channel outputs, do not need to be stored by some internal memory inside the motor organs --- they are in fact ``memorized'' by the limb simply remaining where it is in the absence of a signal.\footnote{Note that we assume that the dynamics of the limbs are fast enough with respect to $T$ that the given distance can always be covered; this reflects one of the main ideas behind this work, namely that for the homing-in phase, which is exceptionally long compared with the distance covered, the limb impedance is not the limiting factor.}
The optimal procedure is thus incremental and optimal at each step and is achieved on-line without memory.

For completeness, we finally check optimality in~\eqref{eq:opti_equal} by evaluating the distortion; we also provide the closed form expression for $\alpha_i$.
\begin{theorem}
\label{thm:dist}
The quadratic distortion $D_i = \mathbb{E}[(\mathbf{A}-\widehat{\mathbf{A}}_{i})^2]$ decreases geometrically in the number of iterations $i$:
\begin{align}
D_{i} = \frac{\sigma_0^2}{(1 + P/N)^{i}}.
\end{align}
The scaling factor $\alpha_i$ increases geometrically in the number of iterations $i$:
\begin{align}
\alpha_i = \alpha_0 (1 +P/N)^{i/2},
\end{align}
where $\alpha_0 = \frac{\sqrt{P}}{\sigma_0}$.
\end{theorem}
The capacity $C$ is exactly achieved and the distortion decreases geometrically (divided by $1 + P/N$) at each iteration. 
An implementation of the scheme, based on the equations given in the theorems, is presented in the bottom left panel of Fig.~\ref{fig:it_model}, and a tentative mapping to physiological components is given in the bottom right panel of Fig.~\ref{fig:it_model}. 

The formulas for the distortion and the general scheme have been previously obtained in an information-theoretic context by Gallager and Nakibo\u{g}lu~\cite{gallager2010} who discussed an older scheme by Elias~\cite{elias1957}. To our knowledge, the constructive proofs of the Elias scheme given here, as well as its application to model goal-directed movements is novel.

\subsection{From a Discrete Time Model to Continuous PVPs}
We presented a discrete-time scheme in line with the ideas of iterative corrections. If the iteration time is constant and equal to $T$, and $n$ iterations have been completed, then the duration of the second phase $t$ is given by $t = nT$. 
The profile of variance is then a discrete set of points at time $\lbrace t_k = kT \rbrace_{k=1}^n$, where variance amplitudes are given by Theorem~\ref{thm:dist}.

When we operationalize the model through the monitoring of PVPs, we note:
\begin{itemize}
\item That $T$ is constant might be true only on average. Different participants may have different values of $T$. Given a single participant, $T$ may also vary with fatigue, learning etc. 
\item Even if $T$ were constant, the uncertainty associated with determining the starting time (and as a result $\tau$), would induce iterations that are not synchronized. 
\end{itemize}

The iterative corrections are thus in practice desynchronized. As a result, when considering sufficiently many trajectories during the construction of a PVP, it is likely that any time interval, however small, will contain at least one new correction from one of the trajectories.
Asymptotically (i.e. considering infinitely many trajectories), we can therefore assume an infinitesimal feedback time $\delta T$, leading to a continuous time formulation. This also means the number of iterations $n = t/\delta T$ goes to infinity. This does not mean that each trajectory is obtained as the result of a continuous on-line control; we posit that a set of trajectories can be described by a continuous model where variance decreases smoothly over time.

With $n = t/\delta T$ one can rewrite~\eqref{eq:opti_equal} as 
\begin{align}
t = \frac{1}{ C'}\log_2 \frac{\sigma_0}{\sigma(t)} \label{eq:second_phase_cont} 
\end{align}
where the sample mean squared error (continuous time) $\sigma^2(t)$ is taken for the ensemble mean squared error (discrete time) $D_n$ and $C' = C/\delta T$ is the capacity in \emph{bit per second}. Note that this local formulation holds for any time interval in the second phase: For arbitrary $\Delta t \geq 0$
\begin{align}
\log_2 \sigma(t + \Delta t) = \log_2 \sigma(t) - C'\Delta t,\label{eq:logscale}
\end{align}
as long as $t$ and $t + \Delta t$ are in the second phase.
In the remainder of this work, we do not discriminate between $C'$ and $C$ and use invariably $C$ for simpler notations.


The remainder of this work is almost entirely empirical: we first validate the theoretical results just presented on the variance decreasing phase. Then, we propose an empirical study of the first phase, since no useful information-theoretic scheme can predict an increasing variance. 
Fitts' law is then derived by combining the results obtained on both phases.

\section{Method}
\label{sec:method}
\subsection{Dataset Presentation}
\label{sub:datasets}
Many studies of goal-directed movements having previously been conducted, we have re-analyzed existing datasets to provide empirical support for our model. 

The first is from a study by Guiard \etal{}~\cite{guiard2011}, following a discrete protocol~\cite{fitts1964}. The task of the participant was to move a cursor from a given starting position to a given line $\D{} = 150$\,mm away, following 5 different instructions: $\#1$ maximize speed (U-Fast), $\#2$ emphasize speed (Fast), $\#3$ balance speed and accuracy (Balanced), $\#4$ emphasize accuracy (Precise), $\#5$ maximize accuracy (U-Precise).

The second is \emph{pointing dynamics dataset} (PD-dataset) by Müller \etal{}~\cite{mueller2017}. The experiment is a replication of Fitts' well-known, one dimensional 1954 reciprocal task~\cite{fitts1954}, where input is tracked by a mouse. Participants have to move back and forth as fast as possible between two targets of same width \W{} and located \D{} apart.
\D{} (in m) $ \in \lbrace{ 0.212,0.353 \rbrace}$ and $\id{} \in \lbrace{ 2,4,6,8 \rbrace}$ (see Eq.~\eqref{eq:fitts}) were fully crossed. 
The PD-dataset is already pre-processed, regularly sampled and denoised with a low-pass filter at about 8Hz. We performed a similar interpolation and low-pass filtering process on the G-dataset. Full details are available in the original publications.

In both datasets, position is a one-dimensional real signal. As explained below, the two paradigms provide a general picture of the speed-accuracy tradeoff in precise aimed movements:
	\begin{itemize}
	\item \emph{Manipulating speed and accuracy.} For the PD-dataset, accuracy is an independent variable, manipulated via a visual target of size \W{}. In the G-dataset, the participants do not aim toward a target but to a point in space. They are in charge of balancing the speed and accuracy of movements to conform to the instruction given by the experimenter: both movement time and accuracy are thus dependent measures. 
	\item \emph{Discrete or reciprocal task.}
The G-dataset was acquired using a discrete task~\cite{fitts1964}, where the cursor is repositioned at the start after each movement, whereas the PD-dataset was acquired using a reciprocal task~\cite{fitts1954}.
	\item \emph{Multi-joint Movements.} In both experiments, the movements solicited  several joints:  fingers, wrist, elbow, shoulder, and sometimes even elicited movement from the back. These should be differentiated from single-joint movements (e.g., wrist-only movements~\cite{meyer1988}) because the latter eliminate the redundancy of the degrees of freedom i.e. the issue at the heart of variability in aimed movements~\cite{guigon2008, todorov2002}. 
	
	\end{itemize}

\subsection{Positional Variance Profiles (PVPs)}
\label{sub:pvp_method}
Tracking the evolution of variance over time is achieved via so-called Positional Variance Profiles (PVP). The first step in computing PVPs is to identify individual movements from the time series of the position. This is achieved using a homemade segmentation algorithm, described in Appendix~\ref{app:segment}. 

PVPs are then computed, considering all the trajectories produced for a given condition, using following operations:
\begin{enumerate}
\item The trajectories are synchronized by using the starting time as the new time origin. 
\item Each trajectory is extended by padding with the final value of position so that all movements last the same amount of time (say 2~s). Indeed, while movements may stop at some point, the position signal remains constant as long as there is no movement, just as if the participant paused longer before going on to the next movement\footnote{There are thus practically 3 phases; the 2 phases discussed previously, and a third stationary phase, where position is conserved and nothing happens.\label{ft:phi}};
\item Once the set of trajectories has been synchronized and extended, compute the variance of the position for this set. This time-series representing the evolution of variance over time is what we call a PVP.
\end{enumerate}

Fig.~\ref{fig:traces} displays a set of synchronized trajectories for a participant of the G-dataset, and Fig.~\ref{fig:pvps} displays several example PVPs computed from empirical data. We asserted the unimodality of PVPs for both datasets. $94\%$ (75 out of 80) of PVPs were found unimodal In the G-dataset; $92\%$ (88 out of 96) in the PD-dataset. Non-unimodal profiles had small secondary peaks due to trajectory outliers. Hence, as expected from \S~\ref{sec:background}, variance profiles for both datasets can be considered unimodal.
\begin{figure}[htbp]
\centering
\includegraphics[width=.9\columnwidth,height=5cm]{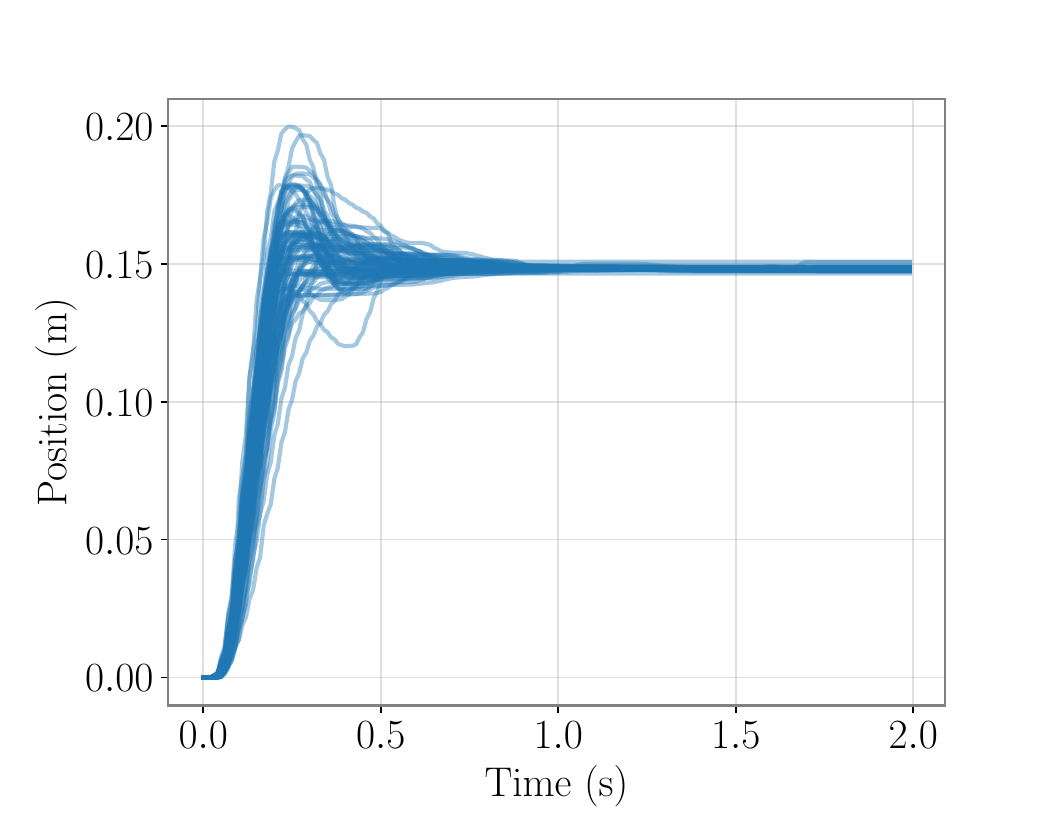} 
\caption{The set of trajectory for Participant X performing under condition $\#3$ (Balanced) for the G-dataset. All trajectories have been time-shifted so as to start at the origin; after about $1~s$, the trajectories are all stationary. 
}
\label{fig:traces}
\end{figure}

\subsection{Calculations performed}
For each dataset, we segmented trajectories as described in \S~\ref{sec:method}. \mt{} was computed (black diamonds on respectively the top and bottom left panel of Fig.~\ref{fig:taustats}), and Fitts' law evaluated (results deferred to \S~\ref{sec:fitts}), to verify that our segmentation procedure produced results in accordance with what is expected from the literature~\cite{soukoreff2004}.

We then computed the PVPs for each block. The decreasing and increasing phases of variances were straightforwardly identified by locating the mode ($\tau, \sigma^2_0$) of the PVP. To evaluate our model, the decreasing phase of variance was fitted with a spline, as described in \S~\ref{sec:empirical_two}.
To study the increasing phase of variance, we systematically observed the effect of the different experimental conditions on $\sigma_0$ and $\tau$, as well as $\D{}_{\tau}$, the distance covered at ~$t=\tau$, since our theoretical model in \S~\ref{sec:model} builds heavily on the position signal.

A box-plot representation of $\tau$ (See Appendix~\ref{app:boxplot}) with data grouped per participant revealed a large disparity in the average $\tau$ per participant (ranging from $\tau = 0.2$ for P2 to $\tau = 0.45$ for P9 in the PD-Dataset). Similarly, the value of $C$, the capacity in our theoretical model, is by definition dependent on each participant.
Hence, we fitter mixed linear model with random intercepts clustered by participant in our regression models when analyzing the PD-dataset, and used one-way Repeated Measures (RM) ANOVA when analyzing the  G-dataset. Greenhouse-Geisser correction, indicated by the $_{GG}$ subscript, was applied when the sphericity assumption was broken.

\section{\added{Exponentially-Decreasing Variance Phase: }Empirical Validation}
\label{sec:empirical_two}

\begin{table*}
\centering
\caption{Mean and Standard deviation Summary for the PD-dataset (left) and G-dataset (right). \label{tab:emp_two}}
\resizebox{.8\textwidth}{!}{
\begin{tabular}{lrrr||lrrr}
\toprule
\W{} ($\D{}\times \id{}$) & $\Omega~(\mu / \sigma)$ & $C~(\mu / \sigma)$ & $r^2~(\mu / \sigma)$ & Instruction & $\Omega~(\mu / \sigma)$ & $C~(\mu / \sigma)$ & $r^2~(\mu / \sigma)$\\ 
\midrule
0.8 ($212 \times 8$) & 1.378/0.214 & 6.22/1.18 & 0.98/0.01 & 1 --- U-Fast & 0.432/0.145 & 4.274/3.18 & 0.87/0.13\\ 
1.4 ($353 \times 8$) & 1.373/0.172 & 6.41/0.91 & 0.99/0.00 & 2 --- Fast & 0.713/0.206 & 4.20/1.51 & 0.94/0.05\\
3.3 ($212 \times 6$) & 1.050/0.195 & 6.95/0.97 & 0.99/0.01 & 3 --- Balanced & 0.851/0.249 & 5.43/1.61 & 0.96/0.05\\
5.5 ($353 \times 6$) & 1.077/0.158 & 6.93/1.18 & 0.99/0.01 & 4 --- Precise & 0.999/0.365 & 5.23/0.984 & 0.97/0.03\\
14.1 ($212 \times 4$)& 0.814/0.144 &6.15/1.59 & 0.97/0.02 & 5 --- U-Precise & 1.310/0.389 & 5.29/0.814 & 0.98/0.02\\
23.6 ($353 \times 4$)& 0.808/0.141 & 6.83/1.64 & 0.98/0.01\\
70.7 ($212 \times 2$)& 0.429/0.130 & 6.04/3.97 & 0.93/0.11\\
117.8 ($353 \times 2$)& 0.527/0.143 & 4.82/2.76 & 0.95/0.03\\
\midrule
$(\overline{\mu} / \overline{\sigma})$ & 0.932/0.368 & 6.30/2.00 & 0.97/0.04 & & 0.861/0.405 & 4.89/1.86 &  0.95/0.08\\
\bottomrule
\end{tabular} 
}
\end{table*}

Theorem~\ref{thm:dist} and its continuous counterpart~\eqref{eq:logscale} forecast a decrease of the positional standard deviation at an exponential rate $C$. The parameter $C$ is a constant that characterizes the channel, is inherent to each participant, and should therefore be unrelated to external constraints, such as task geometry. This makes for the following predictions:

\noindent\textbf{Prediction 1:} The rate at which the positional standard deviation decreases is exponential.

\noindent\textbf{Prediction 2:} That exponential rate $C$ is constant and does not depend (within limits) on the task parameters (i.e. \D{}, \W{}, \id{}, speed-accuracy strategies).

Both predictions are tested on the PD and G-datasets. It is instructive to look at empirical PVPs in log-lin scale before any statistical analysis, see Fig.~\ref{fig:pvps}; the left panel represents the 8 PVPs for a single participant of the PD-dataset, performing all 8 conditions, the middle panel represents the 5 PVPs for a single participant of the G-dataset, performing all 5 conditions, and the right panel is a single PVP taken from the G-dataset with its corresponding spline fit.
One sees that, as expected from~\eqref{eq:logscale}, the second phase is close to linear and that the slopes are all more or less confounded. This suggests that values of $C$ for each PVP are about equal, suggesting that predictions 1 and 2 hold.
\begin{figure*}[t!]
\centering
\makebox[\textwidth]{
\includegraphics[width=0.35\textwidth]{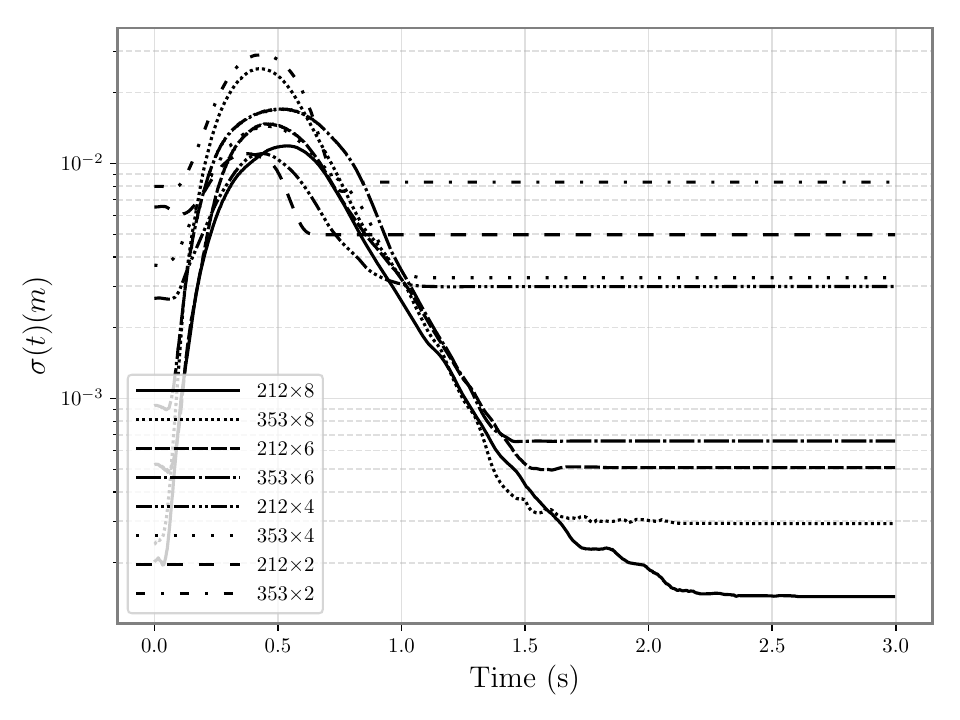} 
\includegraphics[width=0.35\textwidth]{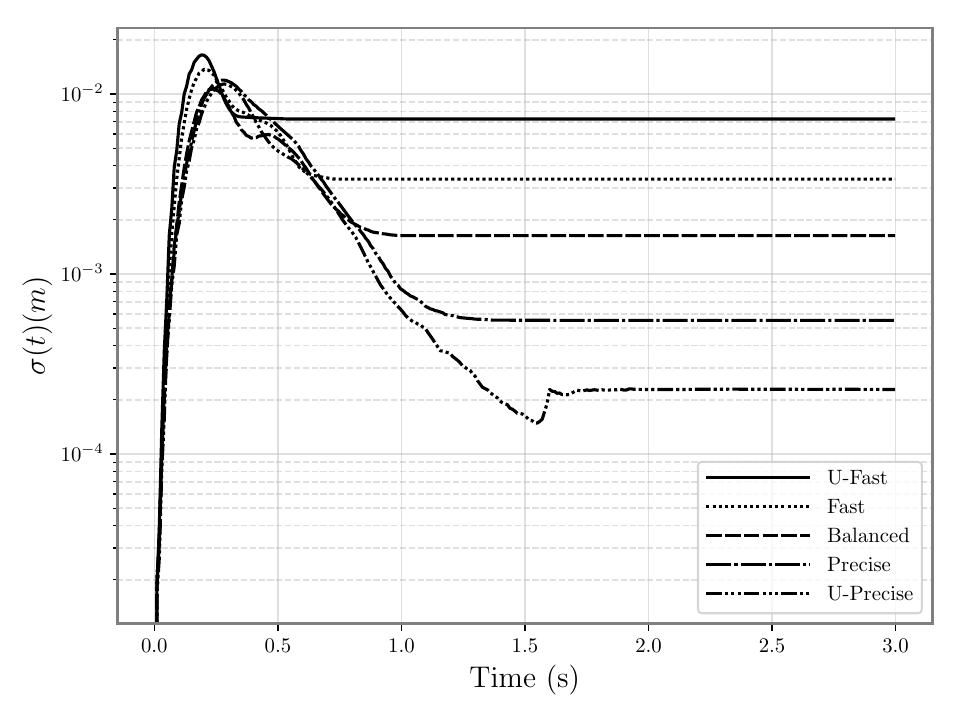} 
\includegraphics[width=0.35\textwidth]{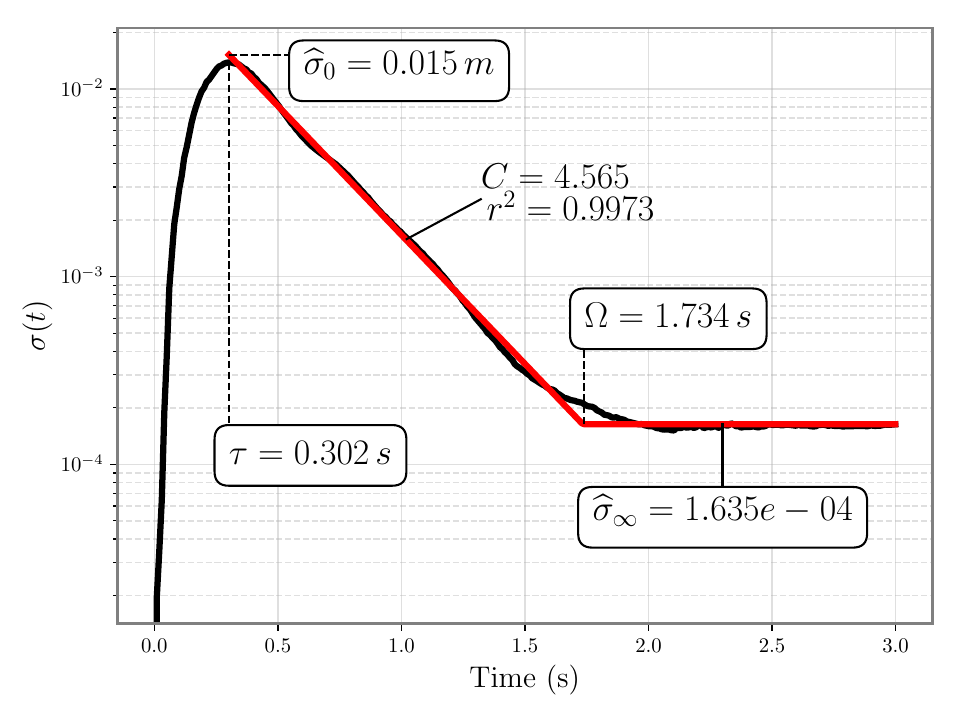} 
}
\caption{Various PVPs in log-lin scale, where time and positional standard deviations are given in standard units. Left panel: PVPs for all conditions of $P8$ from the PD-dataset. Each condition should be read as $\D{}\,(mm) \times \id{}$. Middle panel: PVPs for all conditions of $P8$ from the G-dataset. Right panel: PVP for condition U-Precise ($\#5$) of the G-dataset in log-lin scale with corresponding fit and estimated values of $C$ and $\widehat{\sigma}_{\infty}$.\label{fig:pvps}}
\end{figure*}

\subsection{Prediction 1: Exponential Decrease of Standard Deviation}
As shown in~\eqref{eq:logscale}, an exponential decrease of the standard deviation will appear linear in a log-lin scale. Hence, in a log-lin scale, the PVP is theoretically composed of the following three phases:
\begin{enumerate}
\item a first phase of duration $\tau$;
\item a second phase which decreases linearly, from standard deviation $\sigma_0$ until some value, say $\sigma_{\infty}$;
\item a third stationary phase where standard deviation is constant and equal to $\sigma_{\infty}$, which lasts as long as the duration of the extended trajectories used to construct the PVP.
\end{enumerate}

A piecewise linear model (spline) was fit to the second and third phase of the PVP (log-lin, see the red spline in the right panel of Fig.~\ref{fig:pvps}):
\begin{itemize}
\item a first order polynomial for the second phase. The intercept $\log_2 \widehat{\sigma}_0$ is in theory located at ($\tau,\log_2 \sigma_0$), and the algebraic value of the slope is $-C$;
\item the constant $\log_2 \widehat{\sigma}_{\infty}$ for the stationary phase.
\end{itemize}
The spline has a single knot, defined at ($\Omega, \log_2 \widehat{\sigma}_{\infty}$); equations of the spline are:
\begin{align}
\log_2 \sigma(t) & = \log_2 \widehat{\sigma}_0 - C (t-\tau) \qquad \qquad  \mbox{ if } \tau \leq t \leq \Omega, \\
 & = \log_2 \widehat{\sigma}_0 - C (t-\Omega) = \log_2 \widehat{\sigma}_{\infty}  \qquad \mbox{ else.}
\end{align}
The fit was computed by ordinary least squares (OLS) on the logarithm of the PVP.
This minimization was performed straightforwardly with numerical optimization.

A fit is illustrated in the right panel of Fig.~\ref{fig:pvps}, where the empirical PVP is displayed in black and the piecewise linear fit (spline) is displayed in red. The first phase of the PVP lasts about 300\,ms ($\tau = 302$\,ms); from there the PVP drops quasi-linearly until $\Omega = 1.7$\,s, which defines the knot of the spline at which the PVP levels off. The slope is estimated at $C = 4.6$ bit/second; the endpoint accuracy (stationary phase) is estimated at $\widehat{\sigma}_{\infty} = 1.6e^{-4}$\,m. A coefficient of determination, computed only on the second phase, is also given ($r^2 = 0.9973$).\footnote{It can seem surprising that we give a $r^2$ value computed on the second phase, whereas the model was fit simultaneously on the second and third phase. But notice that by construction, the positional variance is necessarily constant after all movements have terminated. If one were to fit a linear model on the constant phase, one would invariably get a null slope and a perfect fit ($r^2 = 1$), so that the $r^2$ would increase mechanically with longer extension times.}

The coefficients of determination were computed on all PVPs, see Table~\ref{tab:emp_two}. 
For the PD-dataset, the quality of fits are mostly equivalent, except for the 2 conditions with $\id{} = 2$. 
Generally, all goodness of fits computed are high, with an average $r^2$ for the G-dataset of $0.95$ and an average $r^2$ for the PD-dataset of $0.97$.
For the G-dataset, the quality of fit increases with the precision requirement: the mean value of $r^2$ increases and its standard deviation decreases when participants are instructed to emphasize accuracy. 

This can be explained by observing the left and middle panels of Fig.~\ref{fig:pvps}, which show that for the conditions that require low precision (Conditions $\#1$ and $\#2$ in the G-dataset and $\id{}=2$ in the PD-dataset), the second phase is short and the edges of the PVPs are rounded off (likely due to the regularity of human produced movements). We suppose that this rounding off of the edges has more effect the shorter the second phase, leading to a degradation of the $r^2$.

\subsection{Prediction 2: Invariance of $C$}
\paragraph{Effect of \D{} and \W{}  on $C$} 
Average and standard deviation values of $C$ are given Table~\ref{tab:emp_two}; values of $C$ are between 6 and 7 bit/second except for the condition $\W{} = 117.8$. The standard deviations are particularly high for $\id{} = 2$ (coefficient of variation $> 1/2$), indicating that the estimates of the slopes are unreliable in that condition.

A mixed linear regression with random intercept was computed for $C$ with fixed effects \D{} ($\mu = 12.7, \sigma = 23.5, \mbox{cv} = 1.85$), \W{} ($\mu = 2.76, \sigma = 2.99, \mbox{cv} = 1.08$) and interaction term \D{}:\W{} ($\mu = -83.7, \sigma = 72.1, \mbox{cv} = 0.86$) to estimate the effect of each factor, all giving coefficients of variations (cv) close to 1, except for the fixed intercept ($\mu =5.9, \sigma = 0.93, \mbox{cv} = 0.15$). 

\paragraph{Effect of Instructions on $C$}
Average and standard deviation values of $C$ are given Table~\ref{tab:emp_two}. Average values of $C$ are between $4.20$ and $5.43$ bit/second, the maximum being attained for the balanced condition ($\#3$). Instructions to emphasize accuracy lead to lower standard deviations for $C$: from $\sigma(C) = 3.18$ for the condition U-Fast, down to $\sigma(C) = 0.814$ for the condition U-Precise. 
Hence, more precise movements lead to more reliable estimates, likely due to the increase of the duration of the second phase.

To evaluate the effect of instruction on $C$, a one way repeated measures ANOVA was computed ($F(1.88,28.25)=22, \quad p_{GG}=0.17, \quad \eta^2_g = 0.08$).

Both analyses show that these two datasets present very little to no effect of instruction or \D{} and \W{} on the mean value of $C$, thereby supporting the hypothesis that $C$ is a constant that does not depend on the task parameters. 
However, reliable estimation of $C$ seems possible only when the accuracy requirement of the task is stringent.

\subsection{Gaussian Trajectories}
The Gaussian assumption i.e. that observed samples at any given time follow a Gaussian distribution is central to our model.
We asserted the Gaussianity of the position signal throughout the movement for both datasets by:

\begin{enumerate}
\item repeating the first three steps of the PVP construction;
\item starting from $t = \tau$, pooling position data into a vector for each participants for each condition and for each timestamp to create a so-called ``position slice'';
\item normalizing each position slice.
\end{enumerate}
The normalized position slices were pooled, from which a quantile-quantile plot (qqplot) was drawn with a standard normal distribution for comparison, see Fig.~\ref{fig:qqplot}. The left panel represents the qqplot for the G-dataset, whereas the right panel represents the qqplot for the PD-dataset. Fig.~\ref{fig:qqplot} shows the empirical distribution to be almost symmetrical, but has heavier tails than the normal distribution in both cases. For about $96\%$ of the data however (in the range where points are less than $2$ standard deviations away from the mean value), the empirical qqplot is an almost exact fit, see Fig.~\ref{fig:qqplot}. The Gaussian assumption thus seems safe.

\begin{figure}
\centering
\makebox[1\columnwidth]{
\includegraphics[width = 0.5\columnwidth, ]{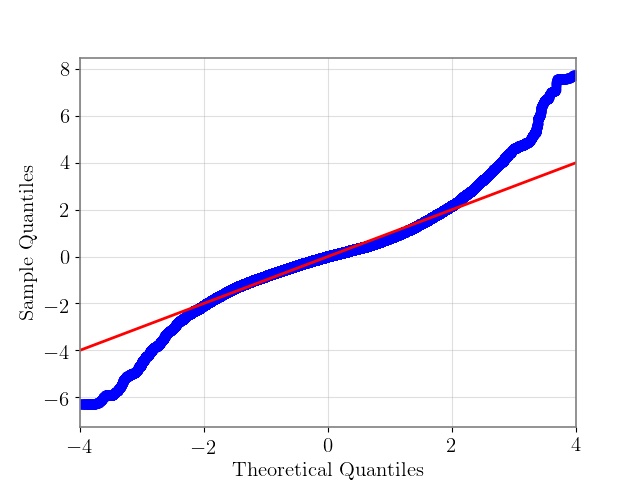} 
\includegraphics[width = 0.5\columnwidth]{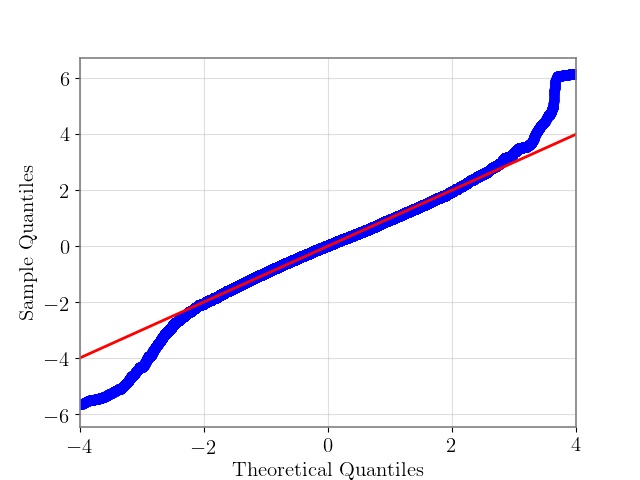} 
}
\caption{Quantile-Quantile plots for the G-dataset (left panel) and the PD-dataset (right panel). Data considered is the reunion of the normalized data per participant and per condition computed for all sampling instants\label{fig:qqplot}}
\end{figure}

\bigskip

While our theoretical model of the second phase shows promising results, an equivalent model for the first phase is required to make sense of the whole movement. As explained previously however, a useful information-theoretic model cannot lead to an increasing variance. Rather than proposing a model for the first phase that draws from a different field, we prefer at this point to present an empirical analysis of the first phase. The surprising result that we find is that the duration of the first phase is almost constant.

\section{Variance Increasing Phase: Empirical Study}
\label{sec:pvp_results}
\D{}, \W{}, \id{} and speed-accuracy strategy were varied and its effect on the features of the first phase $\sigma_0$, $\D{}_{\tau}$ and~$\tau$ was observed.
We find that movements in this first, variance increasing phase are affected primarily by \D. We also find that the duration (and its variations even more) of this first phase is small in comparison to the total movement time.

\begin{figure*}[t!]
\centering
\makebox[\textwidth]{
\includegraphics[width=.35\textwidth]{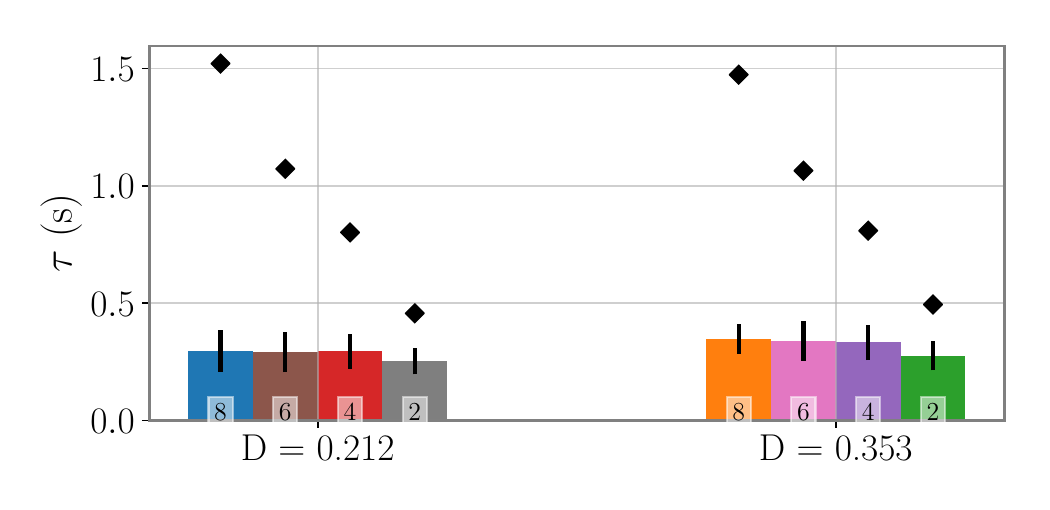} 
\includegraphics[width=.35\textwidth]{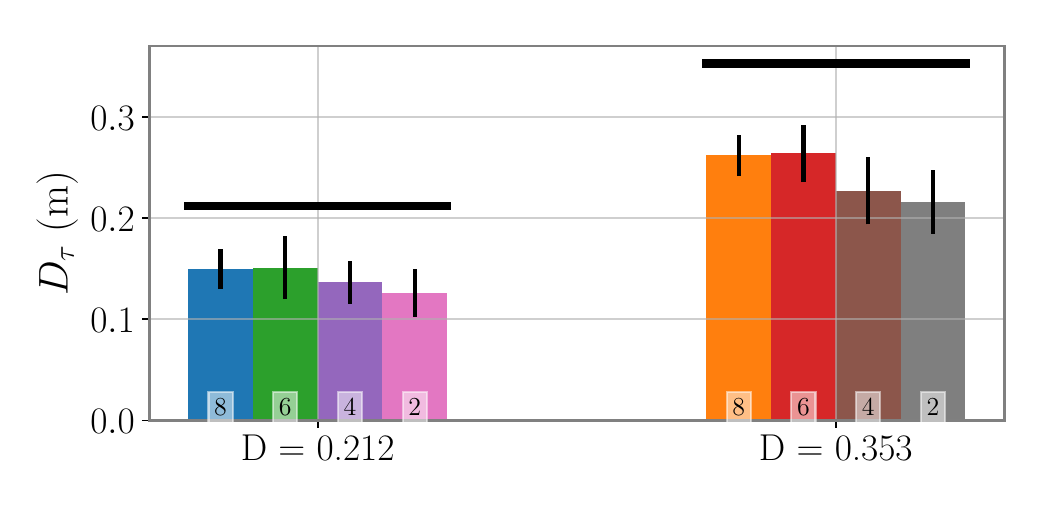} 
\includegraphics[width=.35\textwidth]{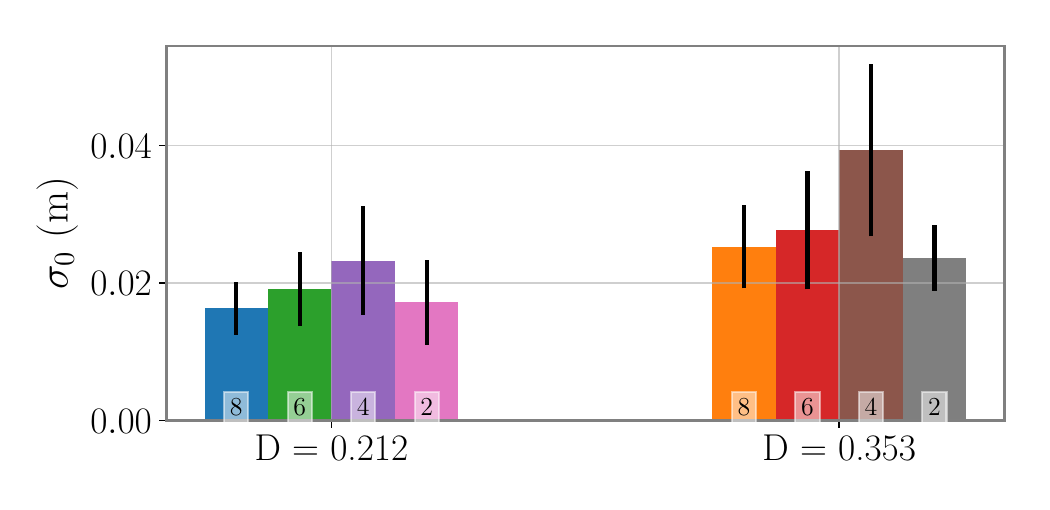}} \\
\makebox[\textwidth]{
\includegraphics[width=.35\textwidth]{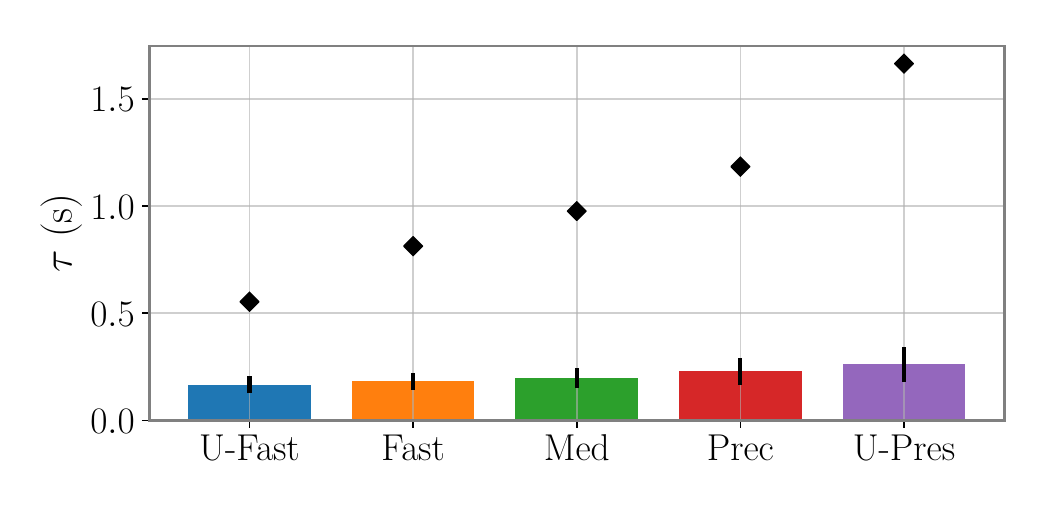} 
\includegraphics[width=.35\textwidth]{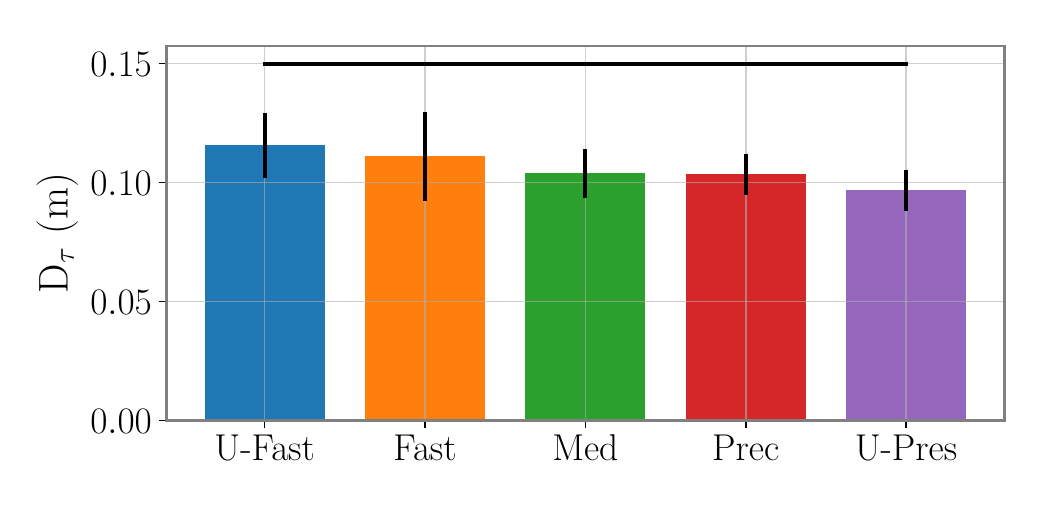} 
\includegraphics[width=.35\textwidth]{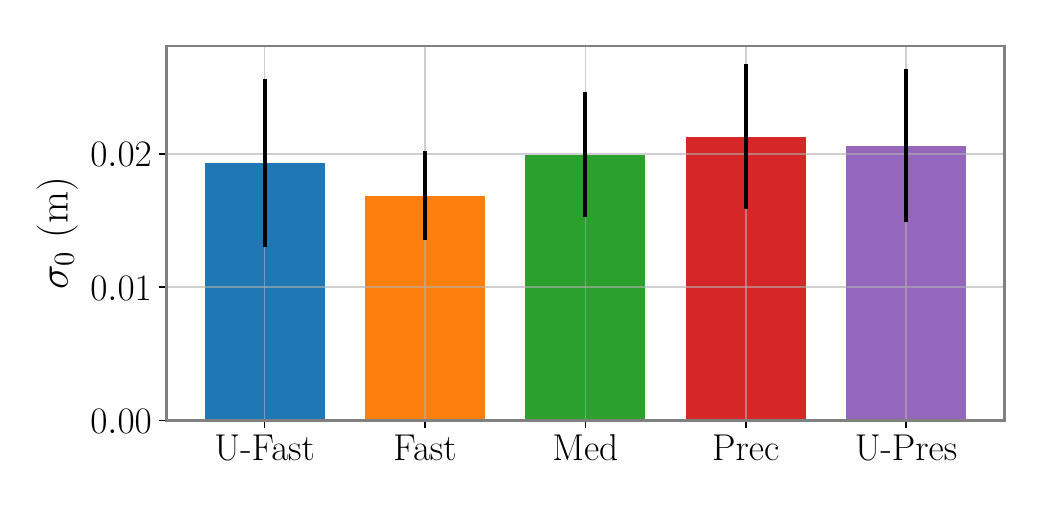} 
}
\caption{Top: Effects of $\D{}$, $\id{}$ and $\W{}$ on $\tau$ (top left panel), $\D{}_{\tau}$ (top middle) and $\sigma_0$ (top right) for the PD-dataset. Bars are grouped by $\D{}$ condition, with the 4 bars on the left corresponding to $\D{}= 0.212~\mbox{m}$ and the 4 bars on the right corresponding to $\D{} = 0.353~\mbox{m}$. Each bar is labeled with its corresponding level of \id{}.
Bottom: Effects of instruction on $\tau$ (bottom left panel), $\D{}_{\tau}$ (bottom middle) and $\sigma_0$ (bottom right) for the G-dataset.
On the $\tau$ panels, total movement time is represented with black diamonds, and on the $\D_{\tau}$ panels, \D{} is represented in a black thick line.\label{fig:taustats}}
\end{figure*}

\begin{table}
\centering
\caption{Mean ($\mu$) and standard deviation ($\sigma$) summary for $\tau$, $\D{}_{\tau}$ and $\sigma_0$ for the PD-dataset. The left column indicates the task conditions, where \W{} and \D{} are expressed in mm. All other units are standard. \label{tab:pdsum}}
\resizebox{\columnwidth}{!}{
\begin{tabular}{lrrrrrr}
\toprule
\W{} ($\D{}\times \id{}$) & $\tau~(\mu / \sigma)$ & $D_{\tau}~(\mu / \sigma)$ & $\sigma_0~(\mu / \sigma)$ \\ 
\midrule
0.8 ($212 \times 8$) & 0.296/0.089 & 0.150/0.020 & 0.016/0.004 \\ 
1.4 ($353 \times 8$) & 0.347/0.063 & 0.262/0.020 & 0.025/0.006 \\
3.3 ($212 \times 6$) & 0.291/0.085 & 0.151/0.031 & 0.019/0.005 \\
5.5 ($353 \times 6$) & 0.340/0.084 & 0.264/0.028 & 0.028/0.009 \\
14.1 ($212 \times 4$)& 0.295/0.074 & 0.137/0.021 & 0.023/0.008 \\
23.6 ($353 \times 4$)& 0.333/0.075 & 0.227/0.033 & 0.039/0.006 \\
70.7 ($212 \times 2$)& 0.252/0.057 & 0.126/0.024 & 0.017/0.006 \\
117.8 ($353 \times 2$)&0.276/0.062 & 0.216/0.032 & 0.024/0.005 \\
\midrule
$(\overline{\mu} / \overline{\sigma})$ &  0.305/0.078 & 0.193/0.059 & 0.024/0.010 \\
\bottomrule
\end{tabular} 
}
\end{table}

\begin{table}
\centering
\caption{Mean ($\mu$) and Standard
deviation ($\sigma$) Summary for $\tau$, $\D_{\tau}$ and $\sigma_0$ for the
G-dataset  \label{tab:gsum}}
\resizebox{\columnwidth}{!}{
\begin{tabular}{lrrrrrr}
\toprule
$\#$ Instruction  &  $\tau~(\mu / \sigma)$ & $D_{\tau}~(\mu / \sigma)$ & $\sigma_0~(\mu / \sigma)$ \\
\midrule
1 --- U-Fast            &  0.167/0.041 &  0.115/0.014 &  0.019/0.006 \\
2 --- Fast            &  0.181/0.039 &  0.111/0.019 &  0.016/0.003 \\
3 --- Balanced   &  0.199/0.046 &  0.103/0.011 &  0.019/0.005 \\
4 --- Precise            &  0.228/0.063 &  0.103/0.009 &  0.021/0.005 \\
5 ---Precise &  0.261/0.081 &  0.096/0.009 &  0.020/0.006 \\
\midrule
$(\overline{\mu} / \overline{\sigma})$ &  0.208/0.065 & 0.106/0.014 & 0.020/0.005 \\
\bottomrule
\end{tabular}
}
\end{table}

\subsection{Effect of $\D{}$ and $\W{}$ on $\tau$}

Values of $\tau$ (see Tab.~\ref{tab:pdsum}) are plotted in the top left panel of Fig.~\ref{fig:taustats}. Results of the regression are given in Table~\ref{tab:tau}, where the estimated mean value/standard deviation, as well as the coefficient of determination, are given for each fixed effect parameter.

\begin{table}[h!tp] \centering 
  \caption{Regression of $\tau$, $\D_{\tau}$ and $\sigma_0$  \label{tab:tau}} 
 \resizebox{\columnwidth}{!}{
\begin{tabular}{lrrrr} 
\toprule
 & $\tau\,(\mu / \sigma)$ & $\D{}_{\tau}\,(\mu / \sigma)$ & $\sigma_0\,(\mu / \sigma)$\\
\midrule
 \W{} & $-$0.667/0.526 & $-$0.256/0.362 & $-$0.012/0.111\\ 
 \D{} & 0.343$^{***}$/0.067  & 0.765$^{***}$/0.046 & 0.079$^{***}$/0.015 \\ 
 \W{}:\D{} & 0.217/1.617 & $-$0.319/1.114 & $-$0.092/0.345 \\ 
 Intercept & 0.225$^{***}$/0.027 & $-$0.014/0.014 & 0.003/0.004 \\ 
  $r^2_m/r^2_c$ & 0.14/0.79 & 0.78/0.82 & 0.28/0.33 \\
\bottomrule
\textit{Note:}  & \multicolumn{2}{r}{$^{*}$p$<$0.1; $^{**}$p$<$0.05; $^{***}$p$<$0.01} \\ 
\end{tabular} 
}
\end{table} 
%


\subsection{Effects of $\D{}$ and $\W{}$ on $D_{\tau}$}
Average values of $D_{\tau}$ are plotted with bars in the top middle panel of Fig.~\ref{fig:taustats}. For comparison purposes, factor level \D{} was also plotted on the top middle panel of Fig.~\ref{fig:taustats} with a thick black line.
Higher levels of \D{} lead to larger $\D{}_{\tau}$:
\begin{itemize}
\item For \D{} = 0.212\,m, $\D{}_{\tau}$ ranges from 0.126\,m ($\id{}=2$) to 0.151~m ($\id{}=6$). The average $\D{}_{\tau}$ over all \id{} levels is 0.141\,m, ($67\%$ of \D{}).
\item For \D{} = 0.353\,m, $\D{}_{\tau}$ ranges from 0.216\,m ($\id{}=2$) to 0.242~m ($\id{}=6$). The average $\D{}_{\tau}$ over all \id{} levels is 0.228\,m ($69\%$ of \D{}).
\end{itemize}

Results of the regression are given Table~\ref{tab:tau}, confirming the effect of \D{} on $\D{}_{\tau}$.

\subsection{Effects of $\D{}$ and $\W{}$ on $\sigma_0$}
Average values of $\sigma_0$ are plotted with bars in the top right panel of Fig.~\ref{fig:taustats}, which shows an increase of $\sigma_0$ with levels of \D{} ($\overline{\sigma}_0 = 0.019\,m$ for $\D{} = 0.212\,m$ and $\overline{\sigma}_0 = 0.029\,m$ for $\D{} = 0.353\,m$).
This is confirmed by the regression in Table~\ref{tab:tau}. 
The sharp ``drop'' in $\sigma_0$ for $\id = 2$, evident from Fig.~\ref{fig:taustats} could indicate that aimed movement in the $\id = 2$ regime are different from other, more precise movements, in line with Crossman's~\cite[Fig.~1]{crossman1983} analysis for low \id{} movements.

\bigskip

Keeping only factors with a great effect size\footnote{As measured by a coefficient of variation much smaller than 1, where the coefficient of variation for a set of samples with observed mean $\mu$ and standard deviation $\sigma$ is defined as the ratio $\sigma/\mu$. There is no difficulty in choosing the cut-off rate, since coefficients of variation obtained are clear cut above or below 1. \added{Significance testing with p-values gives identical results ($\alpha=0.05$).}} 
the regressions of the PD-dataset (Tab.~\ref{tab:tau}) are condensed to:
\begin{align}
& \tau = 0.225 + 0.343\D{} + \varepsilon ; \label{eq:tau} \\
& \D{}_{\tau} = 0.765\D{} + \varepsilon'; \label{eq:dtau} \\
& \sigma_0 = 0.079\D{} + \varepsilon''. \label{eq:sig0}
\end{align}
where $\varepsilon$, $\varepsilon'$ and $\varepsilon''$ are error terms. This shows that $\tau$, $\D_{\tau}$ and $\sigma_0$ of the first phase of the PVP depend primarily on the level of \D{}.

\subsection{Effect of instructions on $\tau$, $\D{}_{\tau}$ and $\sigma_0$}
The average values of $\tau$, $\D{}_{\tau}$ and $\sigma_0$ for the G-dataset are represented respectively on the left, middle and right bottom panels of Fig.~\ref{fig:taustats} and Table~\ref{tab:taug}. 

\begin{itemize}
\item Effect on $\tau$: Grand mean for $\tau$ is $\overline{\tau} = 0.208\,s$ ($\min = 0.167\,\mbox{s}$ for condition U-Fast ($\#1$), $\max =0.261\,\mbox{s}$ for condition U-Precise ($\#5$).
Results of the one-way RM ANOVA with factor Instruction are given Tab.~\ref{tab:taug}, showing a moderate effect ($\eta^2_g = 0.28$) on $\tau$. Instructions to emphasize Accuracy increase the value of $\tau$.

\item Effect on $\D{}_{\tau}$: grand mean for $\D{}_{\tau}$ is $\overline{\D{}_{\tau}} = 0.106\,\mbox{m}$ ($\min = 0.096\,\mbox{m}$ for condition U-Precise $\#5$, $\max = 0.115\,\mbox{m}$ for condition U-Fast $\#1$).
Results of the one-way RM ANOVA with factor Instruction are given Tab.~\ref{tab:taug}, showing a moderate effect ($\eta^2_g = 0.22$) on $\D{}_{\tau}$.
Instructions to emphasize speed lead to larger values of $\D{}_{\tau}$. 
Target distance \D{} is represented in thick black line for comparison on the bottom middle panel of Fig.~\ref{fig:taustats}.

\item Effect on $\sigma_0$: grand mean for $\sigma_0$ is $\overline{\sigma_0} = 0.02\,\mbox{m}$ ($\min = 0.016\,\mbox{m}$ for condition Fast $\#2$, $\max = 0.021\,\mbox{m}$ for condition Precise $\#4$).
Results of the one-way RM ANOVA with factor Instruction are given Tab.~\ref{tab:taug}, showing a very small effect ($\eta^2_g = 0.08$) on \removed{$\D{}_{\tau}$}\added{$\sigma_0$}.
\removed{Instructions to emphasize speed lead to larger values of $\D{}_{\tau}$.}
We also re-ran the ANOVA by excluding the Fast condition ($p_{GG} = 0.44$), giving $\eta^2_g$ close to null.
\end{itemize}

\begin{table}
\centering

\caption{One-way Repeated Measures ANOVA for effect of Instruction on $\tau$, $\D{}_{\tau}$ and $\sigma_0$. \label{tab:taug}}
\resizebox{\columnwidth}{!}{
\begin{tabular}{lllrr}
\toprule
Observation & $F(df_{GG},df_{GG})$ & $p$ & $\eta^2_g$ & $\varepsilon_{GG}$ \\
\midrule
$\tau$  & $F(1.7, 25.48) = 32.4$ & $p_{GG} <  10^{-6}$ & 0.28 & 0.47 \\ 
\midrule
$\D{}_{\tau}$  & $F(2.02, 30.33) = 5.8$ & $p_{GG} = 0.007$ & 0.22 & 0.51 \\ 
\midrule
$\sigma_0$  & $F(2.44, 36.67) = 2.48$ & $p_{GG} = 0.087$ & 0.08 & 0.61 \\ 
\bottomrule
\end{tabular} 
}

\end{table}



\subsection{Discussion}

\paragraph{Variations on $\tau$ can be considered negligible for precise aimed movements.}
\D{} being the primary factor that affects $\tau$, it is informative to compute the actual variations induced by \D{} on $\tau$: For $\D{} = 0.212$\,m, $\tau= 0.298$\,ms, while for $\D{} = 0.353$\,m, $\tau = 0.346$\,ms. 
The two levels of \D{} thus induce about $50\,ms$ variation in $\tau$, which is less than the standard deviation associated with any one condition. 
In comparison, the average \mt{} ranges from about 0.5\,s to 1.5\,s between conditions. Hence, at the level of change induced in \mt{}, variations in $\tau$ are relatively negligible ($5\%$). In other words, the variations in $\tau$ account very little in  explaining the variations in observed \mt{}. 
This comparison is reflected in the top left panel of Fig.~\ref{fig:taustats}, where variations in $\tau$ are hardly visible at the scale of \mt{} (black diamonds).

In the case of the effect of instruction, we similarly observe that variations of $\tau$ are small compared to variations of total \mt{}, see bottom left panel of Fig.~\ref{fig:taustats}, even though the effect of instruction is moderate on $\tau$ ($\eta^2_g = 0.28$). The difference in average $\tau$ it at most about $100\,ms$ i.e. about $10\%$ of the average total movement time. If one excludes the two extreme conditions, under which humans rarely perform in daily life, the difference in $\tau$ comes down to about $5\%$ of the average total movement time, similarly to the PD-dataset.
This suggests that $\tau$ can be considered constant leading to simpler results for a comparatively small loss in modeling power in most practical cases.

The coefficients of determination computed for the mixed effect model on $\tau$ ($r^2_m = 0.14$ and $r^2_c = 0.79$) indicate that most of the variability of $\tau$ is actually due to differences between participants. 
This observation implies that if one is willing to predict $\tau$, models linking $\tau$ to participant characteristics need to be conceived. Conversely, if one if willing to correlate participants characteristics to PVP parameters, $\tau$ should be a good candidate (e.g. older participants can be hypothesized to have, say, a larger $\tau$).

\paragraph{On $\D_{\tau}$ and the isochrony of movements}
The finding that $D_{\tau}$ scales with \D{} is not surprising: Assuming a constant $\tau$, an increase in $D_{\tau}$ translates to an increase in average speed ($\D{}_{\tau}/\tau$). The observation that average speed increases with \D{} is in line with the so-called isochrony of movements~\cite{guiard2009}.
The findings of the G-dataset show that $\D_{\tau}$ (and thus speed) increases with the instruction to emphasize speed (from 0.35\,m/s for $\#5$ U-precise to 0.7\,m/s for $\#1$ U-fast), in agreement with the experimental protocol.
 
In both datasets, the average value of $D_{\tau}$ had the same ratio to \D{}, about $70\%$ ($67\%$ and $69\%$ for two levels of \D{} in the PD-dataset, $71\%$ on average for the G-dataset). Furthermore, most of the variability of $\D_{\tau}$ is captured by the fixed effect model ($r^2_m = 0.78$), indicating low participant differences in $\D{}_{\tau}$. 
It would thus be interesting to test in further work whether this proportion of about $70\%$ continues to hold for larger variations of \D{}, including very short and very long movements.

\section{Synthesis of the two phases: Recovering Fitts' Law}
\label{sec:fitts}
A combination of the two phases provides a full description of the variability of aimed movements. For consistency, it should to both versions of Fitts' law~\eqref{eq:fitts} and~\eqref{eq:fitts_mackenzie}, when applied in Fitts' paradigm.
We intend to show that $\Omega$ and $\widehat{\sigma}_{\infty}$ as defined in Section\S~\ref{sec:empirical_two} (see also see Fig.~\ref{fig:pvps}) can respectively be mapped to \mt{} and $\sigma$ of Fitts' law.

\subsection{Synthesis of the two phases}
Movement duration $\widetilde{\Omega}$ is obtained as the sum of the durations of the two phases:
\begin{align}
\widetilde{\Omega} = \tau + \frac{1}{C} \log_2 \left(\frac{\sigma_0}{\sigma_{\infty}}\right). \label{eq:omega_one}
\end{align}
In \S~\ref{sec:pvp_results}, it was shown that $\sigma_0 \simeq \frac{0.079\,\D{}}{\sigma_{\infty}}$ for the PD-dataset, giving
\begin{align}
\widetilde{\Omega} = \tau + \frac{1}{C} \log_2 \left( \frac{0.079\,\D{}}{\sigma_{\infty}} \right),
\end{align}
and playing with constants, one has
\begin{align}
\widetilde{\Omega} & =  \tau + \frac{1}{C} \log_2 (0.079 \times 4.133)  + \frac{1}{C} \log_2 \left( \frac{\D{}}{4.133\sigma_{\infty}} \right) \\
\widetilde{\Omega} & =  \tau' + \frac{1}{C} \log_2 \left( \frac{\D{}}{4.133\sigma_{\infty}} \right)
\label{eq:fitts_omega_we}
\end{align}
giving an expression close to~\eqref{eq:fitts_mackenzie}.

To obtain the so-called nominal Fitts' law~\eqref{eq:fitts_mackenzie}, we use the Gaussianity of the $\widehat{\mathbf{A}}_i$s, and link the miss rate $\varepsilon$ to the standard deviation of endpoints $\sigma_{\infty} = [2\sqrt{2}\mathrm{erf}^{-1}(1-\varepsilon)]^{-1} W$, where $\mathrm{erf}^{-1}(x)$ is the inverse Gaussian Error Function, and plug this in~\eqref{eq:omega_one} to obtain $\breve{\Omega}$
\begin{align}
\breve{\Omega} & =  \tau + \frac{1}{C} \log_2 \left( 0.079\times 2\sqrt{2}\mathrm{erf}^{-1}(1-\varepsilon) \right) + \frac{1}{C} \log_2 \left(\frac{\D{}}{\W{}}\right)\\
\breve{\Omega} & =  \tau''(\varepsilon)  + \frac{1}{C} \log_2 \left(\frac{\D{}}{\W{}}\right) \label{eq:fitts_omega_w} 
\end{align}

Miss rates $\varepsilon$ are usually small and the standardized methodology considers a constant $4\%$ miss rate. \footnote{In fact, the standard methodology advocated in ISO~\cite{iso9241} uses a constant miss rate of $3.88\%$ (see~\cite{soukoreff2004}, but see also~\cite{gori2018} for a critique), in which case $ 2\sqrt{2}\mathrm{erf}^{-1}(1-\varepsilon) = 4.133$.}

The obtained equations for movement time~\eqref{eq:fitts_omega_w} and~\eqref{eq:fitts_omega_we} are consistent with the nominal and effective Fitts' law~\eqref{eq:fitts} and~\eqref{eq:fitts_mackenzie}, except for the ``$+1$'' term that is missing in formulations for $\widetilde{\Omega}$ and $\breve{\Omega}$. 
This term has been discussed several times in the literature~\cite{soukoreff2004, hoffmann2013} and is actually of little interest --- the changes it induces are sensible only for very low values of the ratio \D{}/\W{}, where Fitts' law is known to be a poor model~\cite{crossman1983}, and where the estimation of $C$ is unreliable anyway.

We evaluated Fitts' law on the set of extracted $\Omega$ and $\widehat{\sigma}_{\infty}$, see Fig.~\ref{fig:glob_fitts}. 
Parameters obtained this way should be compared to those of $\widetilde{\Omega}$ and $\breve{\Omega}$, evaluated from \eqref{eq:fitts_omega_we} and~\eqref{eq:fitts_omega_w} for the G-dataset as well as the PD-dataset.

\begin{figure}
\centering
\makebox[\columnwidth]{
\includegraphics[width=.5\columnwidth]{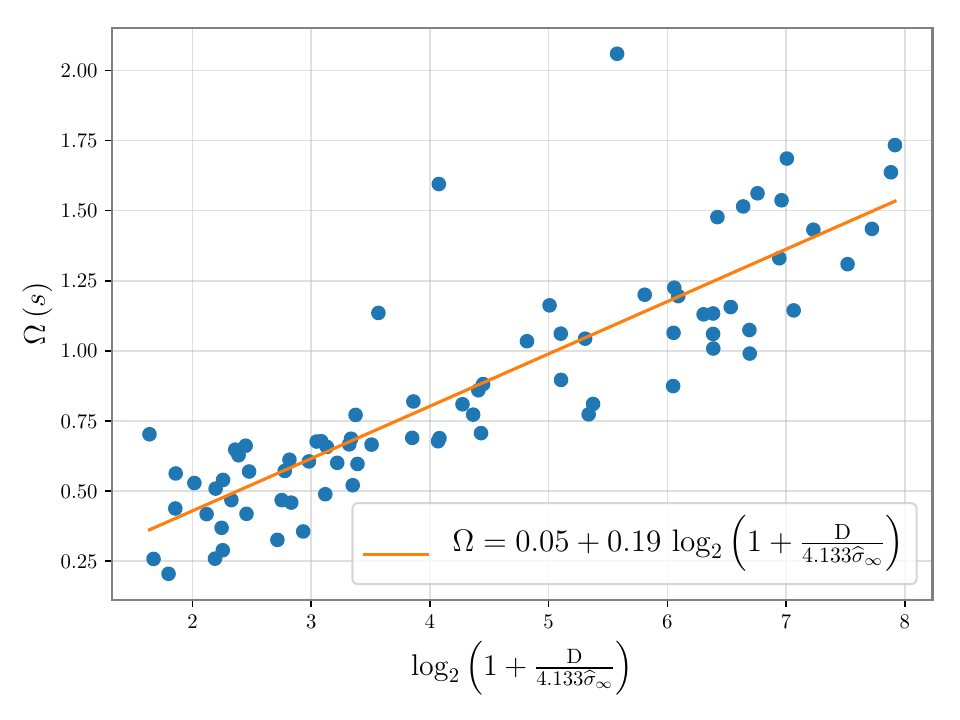} 
\includegraphics[width=.5\columnwidth]{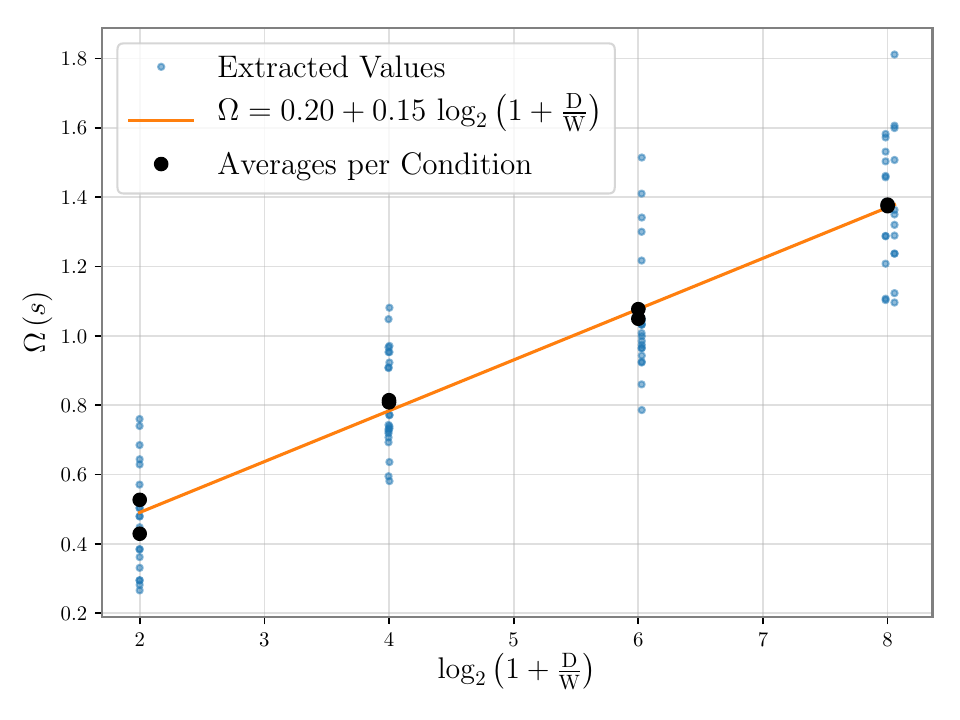} 
}
\caption{Fitts' law evaluated for extracted values of $\Omega$. Left panel: results of Fitts' law~\eqref{eq:fitts_mackenzie} fitted on data of the G-dataset. Right panel: results of Fitts' law~\eqref{eq:fitts} fitted on data of the PD-dataset. The black dots are the per-condition averages, as typically used for regression in most Fitts' law studies (see~\cite{soukoreff2004}). \label{fig:glob_fitts}}
\end{figure}

\paragraph{PD-dataset}
The overall miss rate is $\varepsilon = 2.86\%$. Using the average value of $C = 6.3$ (Table~\ref{tab:pdsum}), one has

\begin{align}
\breve{\Omega}  = 0.07 + 0.16 \log_2 \left(\frac{\D{}}{\W{}}\right).
\end{align}

For comparison, the fitted Fitts' law using the measured $\Omega$ gives $\Omega = 0.2 + 0.15\,\id{}$, see the left panel of Fig.~\ref{fig:glob_fitts} ($r^2_e = 0.81,~r^2 = 99$)\footnote{$r^2_e$ and $r^2$ have been defined before in this paper.}.

\paragraph{G-dataset}
With the average value of $C =4.89$, one has

\begin{align}
\widetilde{\Omega}  = -0.12 + 0.20 \log_2 \left( \frac{\D{}}{4.133\sigma_{\infty}} \right).
\end{align}

For comparison, the fitted Fitts' law using $\Omega$ and $\widehat{\sigma}_{\infty}$ gives $\Omega = 0.05 + 0.19\,\id{}$ ($r^2 = 0.74$), see the right panel of Fig.~\ref{fig:glob_fitts} .

It thus appears that, while the intercept is a little bit underestimated, the estimation of the slope is in line with empirical data. A more precise estimate for the intercept could be obtained by using a more refined model for the first phase. However, this would require many more empirical investigations, and possibly a model of some sort, which is out of the scope of this paper. A second point is that the spline does not factor in the observation that around the maximum variance point, the profiles are smooth, which further induces a bias in the intercept calculation. Furthermore, the most useful parameter in Fitts' law is the slope and not the intercept~\cite{zhai2004a}.

\subsection{Traditional Fitts' law evaluation}
A comparison with Fitts' law parameters \eqref{eq:fitts_empiric_g} and~\eqref{eq:fitts_empiric_pd} obtained only from endpoints, with the traditional methodology confirms the consistency of our novel method.

\paragraph{PD-Dataset} Grand mean for movement time across all conditions was $\overline{\mt{}} = 0.96$\,s. Fitts' law parameters~\eqref{eq:fitts} were estimated via simple linear regression $(r^2_e = 0.742, \quad r^2 = 0.989)$
\begin{align}
\mt{} = 0.13 + 0.17\,\log_2 \bigl(1 + \tfrac{\D{}}{\W{}} \bigr), \label{eq:fitts_empiric_pd}
\end{align}
where $r^2_e$ is the coefficient of determination obtained from all \mt{}s, and $r^2$ the coefficient of determination obtained after averaging \mt{} for each condition.\footnote{The $r^2$ is included for comparison purposes with figures from the literature--- most Fitts' law studies average data prior to regression and usually only $r^2$'s are given.}

\paragraph{G-Dataset} Grand mean for movement time across all conditions was $\overline{\mt{}} = 1.04$\,s. Fitts' law parameters were obtained by simple linear regression $(r^2_e = 0.69)$\footnote{The averaged coefficient of determination $r^2$ cannot be computed here as there is a different level of $\id{}_e$ for each block, even between blocks performed in the same condition.} on~\eqref{eq:fitts_mackenzie} 
\begin{align}
\mt{} = 0.28 + 0.17\,\log_2 \bigl(1 + \tfrac{\D{}}{4.133 \sigma} \bigr). \label{eq:fitts_empiric_g}
\end{align}

As in the previous case, slopes are accurately predicted and intercepts are a bit underestimated, showing the validity of our method compared to the traditional approach.

\section{Conclusion}
\label{sec:conclusion}
The outcome of our theoretical model, Formula~\eqref{eq:logscale}, is a large improvement over the existing quantitative descriptions of the speed-accuracy tradeoff. It relates accuracy to speed \emph{throughout} the movement, and actually results from a joint minimization of speed and accuracy (see Theorem~1 and~\eqref{eq:opti_equal}).
This is in contrast to e.g. Fitts' law that can only describe endpoints, or optimal control models for goal-directed, which often require the time horizon to be preselected~\cite{todorov1998} (but see \cite{harris1998,guigon2008,berret2016,tanaka2006}) and which function with various cost functions~\cite{todorov1998, qian2013, flash1985} that can be considered ad-hoc.
It also has a direct operational implication: the variance decreases exponentially at a constant rate $C$ during the second phase, and this is what largely predicts the movement time.

As discussed by Tanaka \etal{}~\cite{tanaka2006}, two movement expressed as the solutions of an optimization problem where 1) the duration of movement is fixed and the variability is minimized, or 2) the variability is fixed and the duration of movement is minimized are not \textit{a priori} equivalent\footnote{Interestingly, these two partial optimizations correspond to the two different empirical paradigms of Schmidt and Fitts, see e.g.~\cite{plamondon1997}}. The minimization of mutual information that we propose in this work, jointly optimizing for speed and accuracy, actually makes that equivalence.

The functions $\mathbf{f}$ and $\mathbf{g}$ were not defined beforehand, but left unspecified, to be determined as part of the optimization problem.
It turns out that the resulting optimal scheme is linear, surprisingly simple and does not require memory beyond the obvious fact that the limb maintains its position in the steady state. This finding provides support for the many linear schemes found in the motor control literature.

As a final observation,\removed{ it is worth mentioning the historical satisfaction that comes from deriving Fitts' law from an information-theoretic scheme:} \added{ let us recall that} originally Fitts' law was conceived via a vague analogy with Shannon's capacity formula~\cite{shannon1948}, since then deemed flawed (see e.g.~\cite{gori2018} for a brief history).
The information-theoretic framework appears naturally here to solve the aiming problem, upholding that information theory can be a useful tool to model human performance~\cite{chan1990a,sheridan1974}.
In addition, the relevant signal never actually needs to be transformed into the vague notion of \id{}~\cite{guiard2011b}, measured in bits; mutual information is here expressed as a tangible ratio of two spatial variances.

A final advantage comes from the method of PVPs described here. Segmenting movements correctly is notoriously hard~\cite{teasdale1993}. Specifically, determining when a movement (or submovement) ends is particularly difficult, since movements performed with high levels of accuracy usually end very smoothly.\footnote{In fact a movement never truly ends, as keeping the position stationary over time requires a control of some sort.} The method of PVPs dispenses the experimenter of the complicated task of precisely determining stopping times. It also handles trajectories indifferently of the associated control, be it discrete continuous or intermittent.

\bigskip

This work provides a new theoretical model for voluntary movements, based on the study of positional variance profiles (PVPs), which takes into account: (a)~the variability of human produced movements; (b)~a feedback mechanism, essential for reliable aiming; and (c)~intermittent control that becomes continuous at the limit.  
Empirical evidence shows that multi-joint goal-directed movements lead to unimodal PVPs: a first variance-increasing phase of approximately constant duration is followed by a variance-decreasing second phase that lasts until an appropriate accuracy level is reached. We established that:
	\begin{enumerate}
		\item The problem of aiming, during the second component, can be reduced to that of transmitting information from a source (position at the end of the first component) to a destination (limb extremity, cursor) over a channel perturbed by Gaussian noise with the presence of a noiseless feedback link. 
		\item  Using an optimal scheme, in the sense that the transmitted information from source to destination is maximized at each step, we showed that positional variance can decrease at best exponentially during the second phase, as summarized by~\eqref{eq:logscale}, at a constant rate $C$.
	\end{enumerate}
	
Beyond the theoretical implications of this work, we believe that our model can have a direct impact on current practices in the evaluation of pointing devices in HCI.
Indeed, this work suggests an alternative to Fitts' law for the evaluation of pointing devices, by:
\begin{enumerate}
\item instructing users to point exactly on a line as fast as they can (i.e. the case when $C$ is the most reliably estimated),
\item computing PVPs and estimating $C$.
\end{enumerate}
$C$ being constant accross speed-accuracy instructions and levels of \D{} and \W{}, there is no need for repetitive measurements over several \D{} and \W{} conditions.
This new method would have several practical advantages over Fitts' law style evaluation:
\begin{itemize}
\item the model provides a description throughout most of the trajectory, whereas Fitts' law can only inform about endpoints, allowing finer analyses (e.g. rather than asking if one group is faster than the other, one can for investigate whether some group composed of a special population displays a reduction in variance that is subexponential);
\item the evaluation of pointing performance with PVPs does not require a predefined width as in Fitts' task; allowing evaluation in new cases (e.g. for users with strong motor impairment, a predefined width is unfeasible~\cite{davies2014});
\item there is no need to determine the (difficultly tractable) stoppage time of a movement; this point may be particularly useful with motor-impaired populations;
\item since the slope can be estimated from a single experimental condition, using PVPs to estimate Fitts' law parameters saves the experimenter two factors (\D{} and \W{}), saving valuable time. More experimental and methodological work needs to be conducted however to understand how variability in the estimation of $C$ can be reduced.
\end{itemize}
Further practical work should determine whether this method is indeed applicable, particularly whether the accuracy of the estimation of $C$ can be further enhanced.

It would also be interesting to observe the effect on the kinematics if one modified some part of the scheme. For example, what would happen were the feedback non-ideal? Can this serve to model the movement of some motor impaired patients e.g. for pathologies which have afflicted the peripheral nervous system? 
Similarly, what would happen if there were no scaling $\alpha_i$'s or if they were poorly ``chosen''? Could this serve to model tremors (e.g. with $\alpha_i$'s much too large)?
\appendix

\section{Appendix}

\subsection{Segmentation Algorithm}\label{app:segment}
The algorithm works in the following steps (using the pre-processed time series as presented in \S~\ref{sec:method}):
\begin{enumerate}
\item Identify time instants $\lbrace{ t_{0,i} \rbrace}_{i= 1}^n$ when position crosses half of the distance between start and target while maintaining a positive speed, thereby identifying $n$ movements\footnote{In the case of the  reciprocal paradigm, this amounts to removing all movements going right to left and keeping only movements going left to right. The other half of movements can be retrieved by inverting the trajectory and performing the exact same operations. In this paper we only keep movements going left to right, to eliminate potential differences between left to right and right to left movements which we are not concerned with.};
\item Compute the velocity profile (from the position profile), normalize it with respect to its maximum value, 
to determine the start of each movement via thresholding (go back in time from $t_{0,i}$ until the normalized velocity reaches, say, $1\%$);
\item Look for ``dwell periods'' after $t_{0,i}$, i.e., intervals where the absolute value of the normalized velocity is below $1\%$ and when the current position is above the one obtained at $t_{0,i}$. The latest instant of the last dwell period is the end of the movement\footnote{Right after this instant, either a new movement begins (in a reciprocal task), or the cursor is moved back to the start position (in a discrete task).}.
\end{enumerate}

\added{\subsection{$\tau$ box-plot for each participant \label{app:boxplot}}
The Box plot for $\tau$ for each participant is displayed Fig.~\ref{fig:boxplot}.
\begin{figure}[H]
\centering
\includegraphics[width = .8\columnwidth, trim = {0, 1.5cm, 0, 2cm}, clip]{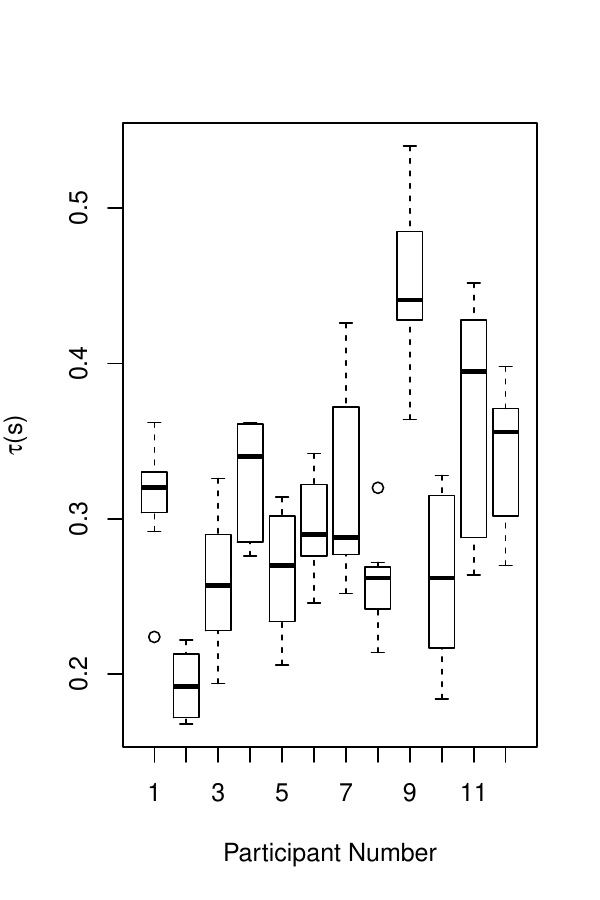} 
\caption{$\tau$ (s) plotted for each participant of the PD-dataset \label{fig:boxplot}}
\end{figure}
}
\subsection{Mathematical Proofs}\label{app:proofs}
In the proofs, we use the natural logarithm ($\log$) rather than its base-2 form ($\log_2$), for simplicity and consistency with information-theoretic textbooks, and is equivalent barring the change of units from bits ($\log_2$) to nats ($\log$)~\cite{cover2012}.
\begin{proof}[Theorem~\ref{thm:ineq}]
We use well-known ingredients from information-theory~\cite{cover2012}. For inequality (\ref{eq:opti_inequal}a):
\begin{align}
I(\mathbf{A};\widehat{\mathbf{A}}_n) &= H(\mathbf{A}) - H(\mathbf{A} | \widehat{\mathbf{A}}_n) \label{eq:1.1} \\
	&= H(\mathbf{A}) - H(\mathbf{A}-\widehat{\mathbf{A}}_n | \widehat{\mathbf{A}}_n) \label{eq:1.2}\\
	& \geq H(\mathbf{A}) - H(\mathbf{A}-\widehat{\mathbf{A}}_n) \label{eq:1.3}\\
	& \geq H(\mathbf{A}) - \frac{1}{2}\log \left( 2 \pi e \mathbb{E}[ (\mathbf{A} - \widehat{\mathbf{A}}_n)^2]\right) \label{eq:1.4}\\
	&= \frac{1}{2}\log \frac{\sigma^2_0}{D_n},\label{eq:1.5}
\end{align}
where~\eqref{eq:1.1} is by definition of mutual information; \eqref{eq:1.2} because of the conditioning by $\widehat{\mathbf{A}}_n$; \eqref{eq:1.3} because conditioning reduces entropy; \eqref{eq:1.4} because the Gaussian distribution maximizes entropy under power constraints and by the entropy formula for a Gaussian distribution; \eqref{eq:1.5} by definition of the distortion and the entropy formula for a Gaussian distribution.

\noindent For inequality (\ref{eq:opti_inequal}b):
\begin{align}
I(\mathbf{A};\widehat{\mathbf{A}}_n) & \leq I(\mathbf{A};\mathbf{Y}^n) \label{eq:2.1} \\
	& = H(\mathbf{Y}^n) - H(\mathbf{Y}^n|\mathbf{A})  \label{eq:2.2} \\
	& = \sum\nolimits_{i=1}^n \left[ H(\mathbf{Y}_i|\mathbf{Y}^{i-1}) - H(\mathbf{Y}_i|\mathbf{Y}^{i-1},\mathbf{A}) \right] \label{eq:2.3} \\
	& = \sum \left[ H(\mathbf{Y}_i|\mathbf{Y}^{i-1}) - H(\mathbf{Y}_i|\mathbf{X}_i) \right] \label{eq:2.4} \\
	& \leq \sum \left[ H(\mathbf{Y}_i) - H(\mathbf{Z}_i)\right] \label{eq:2.5} \\
	& \leq \sum \left[ \frac{1}{2} \log (2\pi e (P_i + N)) - \frac{1}{2} \log (2\pi e N)\right] \label{eq:2.6} \\
	& \leq \sum\nolimits_{i=1}^n \left[ \frac{1}{2} \log ( 1 + P_i/N) \right] \leq nC \label{eq:2.7}	
\end{align}
where~\eqref{eq:2.1} is by the data processing inequality~\cite{cover2012} applied the Markov chain  $\mathbf{A}~\longrightarrow~\mathbf{Y}^i~\longrightarrow~\mathbf{g}(\mathbf{Y}^i) = \widehat{\mathbf{A}}^i$; \eqref{eq:2.2} by definition; \eqref{eq:2.3} by applying the chain rule~\cite{cover2012} to both terms; \eqref{eq:2.4} by design of the feedback scheme; \eqref{eq:2.5} because conditioning reduces entropy for the first term a,d by virtue of the AWGN model for the second term; \eqref{eq:2.6} because the Gaussian distribution maximizes entropy and $\mathbf{X}_i$ and $\mathbf{Z}_i$ are independent; \eqref{eq:2.7} by the concavity of the logarithm function.
\end{proof}

\begin{proof}[Lemma~\ref{lm:Cond}]
The proof consists of finding the conditions that make the inequalities in the proof of Theorem~\ref{thm:ineq} equalities.
Equality in~\eqref{eq:1.3} is equivalent to condition~\ref{enum:cn:3_bis}); Equality in~\eqref{eq:1.4} is equivalent to $\mathbf{A} - \widehat{\mathbf{A}}^i$ Gaussian; Equality in~\eqref{eq:2.1} is equivalent to $H(\mathbf{A}|\mathbf{Y}^i) = H(\mathbf{A}|\mathbf{Y}^i,\mathbf{g}(\mathbf{Y}^i)) = H(\mathbf{A}|\mathbf{g}(\mathbf{Y}^i))$, so that $\mathbf{Y}^i~\longrightarrow~\mathbf{g}(\mathbf{Y}^i)~\longrightarrow~\mathbf{A}$ form a Markov chain~\cite{cover2012}, leading to condition~\ref{enum:cn:4}); Equality in~\eqref{eq:2.5} is equivalent to condition~\ref{enum:cn:3}); Equality in~\eqref{eq:2.6} means the $Y_i$'s are Gaussian; Equality in~\eqref{eq:2.7} leads to condition~\ref{enum:cn:2}) by concavity of the logarithm. Finally, $\mathbf{X}_i$ is Gaussian as the result of the sum of two Gaussians $\mathbf{Y}_i$ and $\mathbf{Z}_i$, and so is $\widehat{\mathbf{A}}_i$ as the sum of $\mathbf{A}$ and $\mathbf{A} - \widehat{\mathbf{A}}_i$. This finally yields condition~\ref{enum:cn:1}).
\end{proof}

\begin{proof}[Theorem~\ref{thm:f}]
We start by considering $\mathbf{X}_i = \mathbf{f}(\mathbf{Y}^{i-1}, \mathbf{A})$, which should be independent of $\mathbf{Y}_{i-1},~~\forall i$ by condition~\ref{enum:cn:3_bis}) of Lemma~\ref{lm:Cond}. This implies the decorrelation 
\begin{align}
\mathbb{E}[\mathbf{f}(\mathbf{Y}^{i-1}, \mathbf{A})(\mathbf{Y}_{i-1})] = 0,~~\forall i. \label{eq:decor}
\end{align}
Since $\mathbf{X}_i$ is a function of two Gaussians $\mathbf{A}$ and $\mathbf{Y}^{i-1}$, the conditional expectation $\mathbf{X}_i = \mathbf{E}[\mathbf{X}_i|\mathbf{A}, \mathbf{Y}^{i-1}]$ is linear, hence $\mathbf{X}_i = \alpha_i (\mathbf{A} - \tilde{\mathbf{f}}(\mathbf{Y}^{i-1}))$. Plugging this in~\eqref{eq:decor} makes for a direct application of the orthogonality principle, showing that $\tilde{\mathbf{f}} = \mathbb{E}[\mathbf{A}|\mathbf{Y}^{i-1}] = \mathbf{g}(\mathbf{Y}^{i-1})$.
\end{proof}

\begin{proof}[Theorem~\ref{thm:incr}]
The goal of the proof is to evaluate $\mathbb{E}[\mathbf{A}|\mathbf{Y}^{i-1}]$.
We first use the operational formula from the orthogonality principle $\mathbb{E}[\mathbf{A}|\mathbf{Y}^{i-1}] = \mathbb{E}[\mathbf{A}(\mathbf{Y}^{i-1})^t] \mathbb{E}[\mathbf{Y}^{i-1}(\mathbf{Y}^{i-1})^t]^{-1}\mathbf{Y}^{i-1}$.
Because the channel outputs are independent (conditions~\ref{enum:cn:3}) and~\ref{enum:cn:4}) from Lemma~\ref{lm:Cond}) and input powers are identical (conditions~\ref{enum:cn:2} from Lemma~\ref{lm:Cond}), $\mathbb{E}[\mathbf{Y}^{i-1}(\mathbf{Y}^{i-1})^t]^{-1} = (P+N)^{-1} \mathbb{I}$, where $\mathbb{I}$ is the identity matrix of size $i-1$.
Then, let $\mathbf{A}_i = \mathbf{X}_i/\alpha_i$ be the unscaled version of $\mathbf{X}_i$ and notice that $\mathbf{A} - \mathbf{A}_i = \mathbf{g}(\mathbf{Y}^{i-1})$. As the channel outputs are independent, we immediately have that $\mathbb{E}[(\mathbf{A} - \mathbf{A}_i)Y_i] = 0$, hence $\mathbb{E}[\mathbf{A}Y_i] = \mathbb{E}[\mathbf{A}_iY_i]$.

Combining both results, we get 
\begin{align}
\mathbb{E}[\mathbf{A}|\mathbf{Y}^{i-1}] &= (P+N)^{-1}\mathbb{E}[\mathbf{A}(\mathbf{Y}^{i-1})^t] \mathbb{I}\mathbf{Y}^{i-1} \\
& = (P+N)^{-1} \sum_{j=1}^{i-1} \mathbb{E}[\mathbf{A}\mathbf{Y}_j]\mathbf{Y}_j \\
&=  (P+N)^{-1} \sum_{j=1}^{i-1} \mathbb{E}[\mathbf{A}_j\mathbf{Y}_j]\mathbf{Y}_j \\
&= \sum_{j=1}^{i-1} \mathbb{E}[\mathbf{A}_j|\mathbf{Y}_j]
\end{align}
where 
\begin{align}
\mathbb{E}[\mathbf{A}_i|\mathbf{Y}_i] & = (P+N)^{-1}\mathbb{E}[\mathbf{A}_i\mathbf{Y}_i]\mathbf{Y}_i \\
&= (P+N)^{-1} \mathbb{E}[\mathbf{X}_i/\alpha_i \cdot{} (\mathbf{X}_i + \mathbf{Z}_i)]\mathbf{Y}_i \\
&= \frac{1}{\alpha_i}\frac{P}{P+N} \mathbf{Y}_i
 \label{eq:espcond}
\end{align}
\end{proof}

\begin{proof}[Theorem~\ref{thm:dist}] 
First notice that we can write $D_{i}$ as 
\begin{align}
D_{i} & =  \mathbb{E}[(\mathbf{A} - \mathbb{E}[\mathbf{A}|\mathbf{Y}^i])^2] \\
 &= \mathbb{E}[(\mathbf{A} - (\mathbb{E}[\mathbf{A}|\mathbf{Y}^{i-1}] + \mathbb{E}[\mathbf{A}|\mathbf{Y}_i]))^2] \\
 &= D_{i-1} - \mathbb{E}[(\mathbb{E}[\mathbf{A}|\mathbf{Y}_i])^2]
\end{align}
Using~\eqref{eq:espcond}, one has
\begin{align*}
D_{i} = D_{i-1} - \frac{1}{\alpha_i^2} \frac{P^2}{P+N}.
\end{align*}
Finally, notice that $D_{i-1} = \mathbb{E}[(\mathbf{A} - \mathbb{E}[\mathbf{A}|\mathbf{Y}^{i-1}])^2] = \mathbb{E}[(\mathbf{A}_i)^2] = P/\alpha_i^2$, to see that
\begin{align*}
D_{i} = D_{i-1} \left( 1 - \frac{P}{P+N} \right) = \frac{D_{i-1}}{1 + P/N}
\end{align*}
The closed form for the distortion is obtained by applying this equation recursively, starting from $D_0 = \mathbb{E}[\mathbf{A}^2] = \sigma^2_0$:

\begin{align}
D_{i} = \frac{\sigma_0^2}{(1 + P/N)^{i}}.
\end{align}
Finally, we evaluate $\alpha_i$ (with $\alpha_0 = \frac{\sqrt{P}}{\sigma_0}$)
:
\begin{align}
\alpha_i = \sqrt{\frac{P}{D_i}} = \frac{\sqrt{P}}{\sigma_0} (1 +P/N)^{i/2} = \alpha_0 (1 +P/N)^{i/2}.
\end{align}
\end{proof}


%
%

\bibliographystyle{plain}
\bibliography{smc}


\end{document}